\documentclass[10pt,journal]{IEEEtran}
\IEEEoverridecommandlockouts
\usepackage{multirow}
\usepackage[utf8]{inputenc}
\usepackage{CJKutf8}
\usepackage{comment}
\usepackage{graphicx}
\usepackage{amsmath,amsthm,amsfonts,amssymb}
\usepackage{cite}
\usepackage{bm}
\usepackage{bbm}
\usepackage{url}
\usepackage{cancel}
\usepackage{array}
\usepackage{color,soul}
\usepackage{multirow}
\usepackage{booktabs}
\usepackage[table,xcdraw]{xcolor}
\usepackage{enumitem}
\usepackage{subcaption}

\usepackage[english]{babel}
\usepackage{algorithm}
\usepackage[noend]{algpseudocode}
\usepackage{setspace}

\theoremstyle{plain}
\newtheorem{thm}{Theorem}
\newtheorem{lem}[thm]{Lemma}

\newtheorem{rem}{Remark}

\allowdisplaybreaks[4]

\begin{document}
\title{Unanticipated Adversarial Robustness of \\ Semantic Communication}

\author{Runxin Zhang,
Yulin Shao,
Hongyu An,
Zhijin Qin,
Kaibin Huang
\thanks{R. Zhang and Z. Qin are with the Department of Electronic Engineering, Tsinghua University, Beijing 100084, China.
This work was performed during R. Zhang's visit at The University of Hong Kong.
}
\thanks{Y. Shao, H. An, and K. Huang are with the Department of Electrical and Computer Engineering, The University of Hong Kong, Hong Kong, China.}
\thanks{Correspondence: {ylshao@hku.hk}.}
}

\maketitle

\begin{abstract}
Semantic communication, enabled by deep joint source-channel coding (DeepJSCC), is widely expected to inherit the vulnerability of deep learning to adversarial perturbations. This paper challenges this prevailing belief and reveals a counterintuitive finding: semantic communication systems exhibit unanticipated adversarial robustness that can exceed that of classical separate source-channel coding systems. 
On the theoretical front, we establish fundamental bounds on the minimum attack power required to induce a target distortion, overcoming the analytical intractability of highly nonlinear DeepJSCC models by leveraging Lipschitz smoothness. We prove that the implicit regularization from noisy training forces decoder smoothness, a property that inherently provides built-in protection against adversarial attacks.
To enable rigorous and fair comparison, we develop two novel attack methodologies that address previously unexplored vulnerabilities: a structure-aware vulnerable set attack that, for the first time, exploits graph-theoretic vulnerabilities in LDPC codes to induce decoding failure with minimal energy, and a progressive gradient ascent attack that leverages the differentiability of DeepJSCC to efficiently find minimum-power perturbations. Designing such attacks is challenging, as classical systems lack gradient information while semantic systems require navigating high-dimensional, non-convex spaces; our methods fill these critical gaps in the literature.
Extensive experiments demonstrate that semantic communication requires up to $14$-$16\times$ more attack power to achieve the same distortion as classical systems, empirically substantiating its superior robustness.
\end{abstract}


\section{Introduction}
Semantic communication is an emerging paradigm that aims to shift the goal of wireless systems from faithfully transmitting bits to conveying task-relevant meaning \cite{weaver1953recent,bourtsoulatze2019deep,xie2021deep,shao2024theory,zhang2025towards}.
Unlike classical designs that optimize symbol-level fidelity \cite{shannon1948mathematical}, semantic communication embraces lossy yet meaning-preserving encoding that retains task-critical information \cite{shao2021learning,wu2024deep,bian2025process}.
Enabled by advances in deep learning (DL), deep joint source-channel coding (DeepJSCC) \cite{bourtsoulatze2019deep,farsad2018deep} has emerged as the prototypical realization of the semantic communication paradigm, achieving strong rate-distortion performance, robustness to channel noise, and high compression efficiency across diverse communication and perception tasks \cite{wu2025deep}.
Recently, 3GPP's first 6G technical report (Release~20, TR~22.870 \cite{3gppTR22870}) formally includes two semantic communication use cases, underscoring its role as a key enabler for 6G and AI-native networks.

Despite these advances, semantic communication remains a nascent and largely empirical field with black-box architectures and limited theoretical understanding of reliability and security. These gaps are especially critical in adversarial or safety-sensitive settings, where semantic errors can cause task-level failures beyond bit errors.
For instance, a small adversarial perturbation that alters the meaning of transmitted features could cause a robotic arm to misinterpret an instruction, a vehicle to misidentify a pedestrian, or a drone swarm to lose coordination.
Addressing such risks demands a principled understanding of the vulnerabilities and intrinsic resilience mechanisms of semantic systems.

This paper aims to address a foundational yet underexplored question: {\it Is semantic communication fundamentally more vulnerable to adversarial attacks than classical systems, or does it possess inherent robustness that has been overlooked?}

The prevailing belief in the field holds that semantic communication, due to its reliance on DL, is inherently susceptible to adversarial perturbations \cite{sagduyu2023semantic}. In contrast, conventional systems are protected by mathematically structured channel codes with provable guarantees against random noise. This contrast in design philosophies has led to an intuitive assumption that semantic communication must be more fragile in adversarial environments \cite{sagduyu2023semantic,tang2023gan,zhou2025rome,nan2023physical,hu2022robust}. However, this assumption has not been rigorously tested or theoretically grounded, and the true security behavior of semantic communication remains an open question.

To answer this question, we undertake a systematic comparative analysis that yields several unexpected findings.

\begin{itemize}[leftmargin=0.45cm]
    \item We challenge the prevailing belief that semantic communication is vulnerable to adversarial attacks and reveal a counterintuitive result: semantic communication systems exhibits a previously under-appreciated robustness against adversarial attacks, even without explicit defense mechanisms. That is, the same feature-based, lossy encoding mechanisms that power semantic efficiency can also confer natural resilience to perturbations.
    \item Our investigation reveals a critical blind spot in classical communication research: the adversarial robustness of structured channel codes. Despite decades of work on coding theory under random noise models, the behavior of codes such as low-density parity-check (LDPC) under intelligently crafted, targeted attacks has remained almost entirely unexplored. This paper provides the first systematic analysis of such scenarios, opening the door to a broader, security-aware perspective in coding theory.
    \item We place these findings on a firm theoretical footing through a bounding analysis that addresses the analytical intractability of highly nonlinear DeepJSCC models. By leveraging Lipschitz continuity and first-order Taylor expansion of the decoder, we derive fundamental lower bounds on the minimum attack power required to induce a target distortion, and establish a sufficient condition under which semantic communication provably outperforms classical systems in adversarial settings. The analysis reveals that the built-in resilience of semantic communication is rooted in the smoothness of the decoder, which is implicitly regularized by noisy training of DeepJSCC.
    \item We validate our findings through extensive numerical experiments using a suite of purpose-built attack methodologies that fill critical gaps in the literature. For classical systems, we develop a structure-aware vulnerable set attack that, for the first time, exploits graph-theoretic vulnerabilities (e.g., short cycles, weak constraint subspaces) in LDPC codes to induce decoding failure with minimal energy. For semantic systems, we propose a progressive gradient ascent (PGA) attack that leverages the differentiability of DeepJSCC to efficiently navigate the high-dimensional, non-convex signal space. Experiments on image transmission and massive MIMO CSI feedback corroborate our theoretical analysis and demonstrate that semantic communication consistently requires an order-of-magnitude higher attack power to achieve the same distortion as classical systems, confirming that its inherent resilience is not merely a theoretical possibility but a practically observable property across diverse tasks.
\end{itemize}

\section{Related Work}\label{sec:II}
To establish the necessary background for our comparative analysis, this section examines four interconnected research threads: semantic communication and DeepJSCC, adversarial machine learning, the adversarial robustness of semantic systems, and the classical perspective on intelligent jamming.

\textit{Semantic communication and DeepJSCC}:
The concept of semantic communication traces back to Weaver's extension \cite{weaver1953recent} of Shannon's theory \cite{shannon1948mathematical}, which reframed the goal of communication as conveying meaning beyond bit transmission.
Advances in DL have made this vision practical, with DeepJSCC \cite{bourtsoulatze2019deep,farsad2018deep} emerging as a key technology.
Building on this foundation, subsequent work introduced enabling techniques such as shared knowledge bases \cite{shi2021new,zhou2022cognitive,jiang2022reliable, hu2023robust}, the information bottleneck \cite{tishby2015deep} for distilling task-relevant features \cite{liu2021rate,sana2022learning,shao2021learning,xie2023robust}, and generative semantic communications \cite{grassucci2024enhancing, jiang2024diffsc}. 
These advances are now shaping early 6G research and pre-standardization efforts, such as DeepJSCC-based CSI feedback \cite{guo2022overview,gizzini2020deep}, task-oriented hybrid ARQ \cite{hu2025semharq}, and semantic encoding in the user plane \cite{wang2025explicit}.

\textit{Adversarial DL}:
The vulnerability of deep neural networks (DNNs) to adversarial perturbations is a well-established challenge in DL \cite{szegedy2013intriguing}. Attacks are generally categorized by adversary knowledge: white-box attacks leverage model gradients (e.g., FGSM \cite{goodfellow2014explaining}, PGD \cite{madry2017towards}, C\&W \cite{carlini2017towards}), while black-box attacks rely on model outputs (e.g., transfer- and query-based methods \cite{papernot2017practical,ilyas2018black,brendel2017decision,shi2023reinforcement}). In response, various defenses have been proposed. Adversarial training, which incorporates adversarial examples during optimization \cite{goodfellow2014explaining,madry2017towards,maini2020adversarial,pang2020bag}, remains the most empirically robust. Other approaches include feature denoising \cite{xie2019feature}, model distillation \cite{zi2021revisiting}, and certified defenses such as randomized smoothing \cite{cohen2019certified}. Collectively, this body of work provides essential tools for analyzing the adversarial robustness of DL-based systems.

\textit{Adversarial robustness of semantic communications}:
The DL foundations of semantic communication have naturally led to early studies of its adversarial robustness, based on the assumption that DeepJSCC inherits the vulnerabilities of its underlying DNNs. Under this premise, several works demonstrated the disruptive potential of attacks: \cite{sagduyu2023semantic} applied FGSM perturbations to both source data and channel signals; \cite{tang2023gan} constructed the intelligent attacks based on generative adversarial network (GAN);  \cite{zhou2025rome} proposed a gradient-free perturbation generator capable of producing attacks at controllable power levels. Defenses have also been explored, including adversarial training \cite{nan2023physical,hu2022robust} and ensembling frameworks such as ROME-SC \cite{zhou2024robust}, which uses a dedicated perturbation generator and detector for dynamic defense.

However, the prevailing belief that semantic communication is inherently more vulnerable than classical systems remains unexamined. More critically, the focus on DL vulnerability may have obscured an important, complementary possibility: the defining feature of DeepJSCC, i.e., its end-to-end training for channel noise resilience \cite{wu2025deep}, may materially affect its adversarial behavior. We posit that this noisy-channel training regimen makes DeepJSCC more than just a generic DNN; it potentially confers an intrinsic and unanticipated robustness to attacks. This paper aims to provide the first rigorous test of this hypothesis by directly comparing semantic systems against their classical counterparts.

\textit{Adversarial robustness of classical systems}:
To enable a fair comparison, we must also understand how classical systems behave under attack. Intelligent jamming has primarily been studied from an information-theoretic perspective under the arbitrarily varying channel (AVC) model \cite{csiszar2002capacity}. The capacity of discrete AVCs was established in \cite{csiszar2002capacity}, and extended to continuous channels in \cite{hughes1988capacity}. These elegant capacity results, however, rely on random-coding ensembles and thus do not reveal how structured code families (e.g., LDPC or polar) behave or can be designed under intelligent attacks. This paper bridges that gap through the first systematic coding-theoretic study of structured codes and their iterative decoders under intelligent adversaries, advancing toward attack-resilient designs that embed security as a core principle of reliability.

\section{The Comparative Framework}\label{sec:III}
To rigorously evaluate the adversarial robustness of semantic communication, we establish a controlled environment where semantic and classical systems can be compared on equal footing. This section introduces a framework for assessing adversarial robustness under identical bandwidth, power, and channel conditions.

We consider a canonical source reconstruction task: a source vector $\bm{x} \in \mathbb{R}^M$ is transmitted over a wireless channel and reconstructed as $\widehat{\bm{x}}$, as illustrated in Fig.~\ref{fig:system}. The source dataset is denoted by $\Omega = \{\bm{x}^{(1)}, \bm{x}^{(2)}, \ldots, \bm{x}^{(K)}\}$.

The transmission occurs over a flat-fading channel with coefficient $h$. With transmitter-side phase compensation, the real-valued baseband received signal can be expressed as
\begin{equation}
 \bm{r} = |h| \bm{z} + \bm{\omega},   
\end{equation}
where $\bm{z} \in \mathbb{R}^N$ is the transmitted symbol vector with normalized power, and $\bm{\omega} \sim \mathcal{N}(0, \sigma^2_\omega \bm{I}_N)$ is independent and identically distributed (i.i.d.) additive white Gaussian noise (AWGN). The received SNR is given by $\eta \triangleq {|h|^2}/{\sigma^2_\omega}$.

\begin{figure}[t!]
\includegraphics[width=0.95\columnwidth]{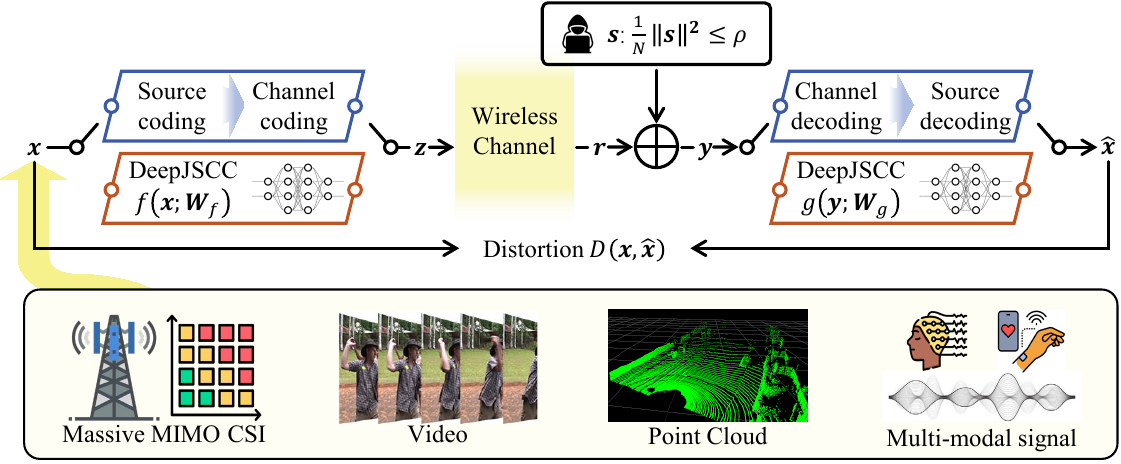}
\caption{A comparative framework for fair evaluation under equal bandwidth, transmit power, and channel conditions. The perturbation $\bm{s}$ is subject to a power constraint $\rho$, and the semantic fidelity is quantified by the distortion $D(\bm{x},\widehat{\bm{x}})$.
}
\label{fig:system}
\end{figure}

\subsection{Classical and Semantic Communications}
In a classical communication system, the transmitter processes source data $\bm{x}$ through separate source and channel coding (SSCC) stages to produce channel symbols $\bm{z}$, as depicted in Fig.~\ref{fig:system}. The receiver reverses these operations to reconstruct $\widehat{\bm{x}}$ from the received signal $\bm{r}$.

By contrast, semantic communication systems employ a DeepJSCC encoder $f(\bm{x}; \bm{W}_f)$ to directly map $\bm{x}$ to channel symbols $\bm{z}$, and a DeepJSCC decoder $g(\bm{r}; \bm{W}_g)$ to reconstruct the source from $\bm{r}$, where $\bm{W}_f$ and $\bm{W}_g$ denote the parameters of the encoder and decoder, respectively. 
DeepJSCC learns both encoder and decoder in an end-to-end fashion to minimize a loss function:
\begin{equation}
\mathcal{D}  \triangleq \frac{1}{K} \sum_{k=1}^{K} D\big(\bm{x}^{(k)}, \widehat{\bm{x}}^{(k)}\big),
\end{equation} 
where $D(\bm{x}, \widehat{\bm{x}})$ is a distortion measure for the specific source.\footnote{This framework is consistent with the principle of remote source coding \cite{witsenhausen2003indirect}, where the receiver reconstructs task-relevant functions of the source, thus covering task-oriented semantic communication \cite{shao2021learning}.}

\subsection{Adversarial Model} 
We consider an adversary that injects a malicious perturbation $\bm{s}\in\mathbb{R}^{N}$ into the channel, yielding a perturbed received signal $\bm{y} \triangleq \bm{r} + \bm{s}$. 
The adversary operates under a white-box assumption: it has full knowledge of the coding schemes (either the SSCC schemes in classical systems or the DeepJSCC encoder and decoder in semantic communication), and can observe the received signal $\bm{r}$ to craft sample-specific perturbations. This represents a worst-case scenario, aligning with adversarial evaluation practices in security and machine learning.

The adversary's objective is to cause significant reconstruction distortion $\mathcal{D}$ while minimizing its own power expenditure $\rho \triangleq \mathbb{E} [\|\bm{s}\|^2]$. Consequently, a natural measure of a system's robustness is the minimum attack power required to provoke a prescribed level of distortion. Formally, for a target distortion $\mathcal{D}^*$, we define
\begin{eqnarray}\label{eq:problem}
      &&\hspace{-1 cm}\rho^* = \min \rho,   \\
      &&\hspace{-1cm} \text{s.t.} \quad \mathcal{D} \geq \mathcal{D}^*. \notag
\end{eqnarray}
A system is deemed more robust if it requires higher attack power $\rho^*$ to achieve the distortion $\mathcal{D}^*$.

\section{Comparative Robustness Analysis}\label{sec:IV}
Building on the comparative framework established in the previous section, we now analyze the robustness of both classical and semantic communication systems to determine whether the prevailing belief that semantic systems are inherently more vulnerable holds true. Our analysis reveals a counterintuitive finding: despite being built upon DL architectures that are often susceptible to adversarial perturbations, semantic communication systems can, and in typical configurations do, exhibit greater robustness against adversarial attacks than their classical counterparts. Through rigorous derivation, we also uncover the underlying mechanisms and conditions that give rise to this unexpected resilience.

In an adversarial setting, the average distortion achieved by a communication system depends jointly on the transmission scheme and the attack model. To clearly distinguish between system configurations, we introduce the following subscript notation:
\begin{itemize}[leftmargin=0.5cm]
    \item Subscripts ``sscc'' and ``sem'' denote classical and semantic communication systems, respectively.
    \item Subscripts ``1'' and ``0'' indicate the presence or absence of adversarial attacks.
\end{itemize}
For instance, $\rho^*_{\text{sem},1}$ denotes the minimum attack power required for the semantic system to achieve a target distortion $\mathcal{D}^*$ under attack, while $\mathcal{D}_{\text{sscc}, 0}$ represents the distortion achieved by a classical SSCC system in the absence of any attack.

With these definitions in place, our objective is to compare $\rho^*_{\text{sem},1}$ and $\rho^*_{\text{sscc},1}$ in the subsequent analysis. This task is nontrivial for two main reasons. 
First, the DeepJSCC encoder and decoder are highly nonlinear DNNs operating over high-dimensional parameter spaces, which complicates analytical treatment. 
Second, the performance of classical communication systems lacks a closed-form expression under adversarial conditions. 
These fundamental structural differences and analytical intractabilities preclude a direct, explicit characterization of the distortion-attack power relationship for either system. 
To overcome this challenge, we adopt a bounding approach that yields sufficient conditions for comparison. Throughout this section, we assume mean squared error (MSE) distortion, i.e., $D(\bm{x}, \widehat{\bm{x}}) = \| \bm{x} - \widehat{\bm{x}} \|^2$.

\subsection{Attack Power Analysis of Semantic Communication}\label{sec:IVA}

In the semantic communication system, the DeepJSCC encoder and decoder are realized by DNNs. We assume the decoder is Lipschitz continuous, a common property in modern DNN architectures, with a Lipschitz constant $G$. Hence, for any two inputs $\bm{r}$ and $\bm{r}^\prime$, we have
\begin{equation}
\|g(\bm{r}; \bm{W}_g) - g(\bm{r}^\prime; \bm{W}_g)\| \leq G \|\bm{r}-\bm{r}^\prime\|.
\end{equation}

Adversarial perturbations are typically small in magnitude. Under this condition, the effect of an attack $\bm{s}$ on the decoder output can be approximated via a first-order Taylor expansion:
\begin{equation}\label{eq:taylor}
    g(\bm{y}; \bm{W}_g) \approx g(\bm{r}; \bm{W}_g) + \bm{\mathcal{J}}_{g}(\bm{r}) \bm{s},
\end{equation}
where $\bm{\mathcal{J}}_{g}(\bm{r}) \in \mathbb{C}^{M \times N}$ denotes the Jacobian matrix of the DeepJSCC decoder evaluated at $\bm{r}$. The attack power is denoted as $\|\bm{s}\|^2 = \rho_{\text{sem},1}$.

To characterize the average distortion $\mathcal{D}_{\text{sem},1}$, we first examine the distortion for a single source sample $\bm{x}$, yielding
\begin{eqnarray}\label{e:D_p_sem_1}
\hspace{-0.5cm}&&  \hspace{-0.4cm}D_{\text{sem},1}  = \left(g(\bm{y}; \bm{W}_g) - \bm{x}\right)^\top \left(g(\bm{y}; \bm{W}_g) - \bm{x}\right) \\
\hspace{-0.5cm}&& \approx \left(   g(\bm{r};\! \bm{W}_{\!g})  \! +  \!  \bm{\mathcal{J}}_{\! g}(\bm{r};\! \bm{W}_{\!g}) \bm{s} \! - \! \bm{x}   \right)^{ \! \top}  \! \left(   g(\bm{r};\! \bm{W}_{\!g})  \! +  \!  \bm{\mathcal{J}}_{ g}(\bm{r}) \bm{s} \! - \! \bm{x}   \right) \notag\\
\hspace{-0.5cm}&& =  \underbrace{\left\|g(\bm{r}; \bm{W}_g ) - \bm{x}\right\|^2}_{D_{\text{sem},0}}  +  \underbrace{\left(\bm{\mathcal{J}}_{ g}(\bm{r}) \bm{s} \right)^{\top}  \bm{\mathcal{J}}_{\! g}(\bm{r}) \bm{s}}_{B}  \notag\\
\hspace{-0.5cm}&&\hspace{4cm}  +  \underbrace{2 (g(\bm{r} ; \bm{W}_g) - \bm{x})^{\top}  \bm{\mathcal{J}}_{\! g}(\bm{r}) \bm{s}}_{C}, \notag
\end{eqnarray}
where $D_{\text{sem},0}$ denotes the distortion in the absence of an attack (i.e., with only channel noise).

Since all terms in \eqref{e:D_p_sem_1} are scalars, we square the expression to relate $B$ and $C$:
\begin{eqnarray}\label{e:D_inequality}
\hspace{-0.3cm}&&\hspace{-0.4cm} C^2 = \big( D_{\text{sem},1}   -   D_{\text{sem},0}   \big)^{ 2}   +   B^2   -   2 \big(D_{\text{sem},1}   -   D_{\text{sem},0} \big) B \\
\hspace{-0.3cm}&& = 4  \big( g(\bm{r}; \bm{W}_g ) - \bm{x}\big)^{  \top}  \bm{\mathcal{J}}_{  g}(\bm{r} ) \bm{s} \left( \bm{\mathcal{J}}_{  g}(\bm{r} ) \bm{s} \right)^{  \top}   \big(g(\bm{r}; \bm{W}_g ) - \bm{x}\big) \notag\\
\hspace{-0.3cm}&& \hspace{-0.4cm} = \!  4 \text{tr}  \left\{  \!   \bm{\mathcal{J}}_{ \! g}(\bm{r} ) \bm{s}    \left(   \bm{\mathcal{J}}_{\! g}(\bm{r} ) \bm{s}   \right)^{ \! \top}   \big(g(\bm{r}; \bm{W}_{\!g} ) - \bm{x}\big) \big(g(\bm{r} ; \bm{W}_g) - \bm{x}\big)^{ \! \top}  \! \right\} \notag\\
\hspace{-0.30cm}&&\hspace{-0.05cm} 
\stackrel{(a)}{\leq} 4 \text{tr}\left\{ \bm{\mathcal{J}}_{  g}(\bm{r} ) \bm{s} \left( \bm{\mathcal{J}}_{  g}(\bm{r} ) \bm{s} \right)^\top \right\}\notag\\
\hspace{-0.30cm}&& \hspace{0.8cm} \text{tr}\left\{ (g(\bm{r}; \bm{W}_g ) - \bm{x}) (g(\bm{r}; \bm{W}_g ) - \bm{x})^\top \right\}  
= 4 B   D_{\text{sem},0}, \notag
\end{eqnarray}
where (a) follows from the fact that both matrices inside the trace are positive semi-definite.

By rearranging the inequality in \eqref{e:D_inequality} and invoking the natural assumption that the distortion under attack is at least as large as the distortion without attack, i.e., $ D_{\text{sem},1} - D_{\text{sem},0} \geq 0$,  we can bound the term $B$ as
\begin{equation}\label{e:range_B}
\hspace{-0.2cm} \sqrt{ \! D_{\text{sem},1}} \! - \! \sqrt{ \! D_{\text{sem},0}} \! \leq \! \sqrt{ \! B} \! \leq \! \sqrt{ \! D_{\text{sem},1}} \! + \! \sqrt{ \!D_{\text{sem},0}}.
\end{equation}

Because \eqref{e:range_B} holds sample-wise, it remains valid after taking expectations over the source and channel realizations. Focusing on the left inequality and applying the expectation operator yields
\begin{eqnarray}\label{e:rho_sem_inequality}
\hspace{-0.3cm}&& \hspace{-0.8cm} \mathcal{D}_{\text{sem},1} = \mathbb{E}_{\bm{r},\bm{s}}[D_{\text{sem},1}] \notag\\
\hspace{-0.3cm}&&  \leq \mathbb{E}_{\bm{r},\bm{s}}[B] + \mathbb{E}_{\bm{r},\bm{s}}[D_{\text{sem},0}] + \mathbb{E}_{\bm{r},\bm{s}}[2\sqrt{B D_{\text{sem},0}}] \notag\\
\hspace{-0.3cm}&& = \mathbb{E}_{\bm{r},\bm{s}}[B] + \mathcal{D}_{\text{sem},0} + \mathbb{E}_{\bm{r},\bm{s}}[2\sqrt{B D_{\text{sem},0}}] \notag\\
\hspace{-0.3cm}&&\hspace{-0.05cm}  \stackrel{(a)}{\leq} \mathbb{E}_{\bm{r},\bm{s}}[B] + \mathcal{D}_{\text{sem},0} + 2\sqrt{ \mathbb{E}_{\bm{r},\bm{s}}[B D_{\text{sem},0}] } \notag\\
\hspace{-0.3cm}&&\hspace{-0.05cm}  \stackrel{(b)}{\leq}  \mathbb{E}_{\bm{s}}  \left[   G^2 \|\bm{s}\|^2 \right]  +  \mathcal{D}_{\text{sem},0}  +  2  \sqrt{  \mathbb{E}_{\bm{r},\bm{s}}  \left[  G^2  \|\bm{s}\|^2  D_{\text{sem},0}  \right] } \notag\\
\hspace{-0.3cm}&&\hspace{-0.05cm}   \stackrel{(c)}{=} G^2 \rho_{\text{sem},1} + \mathcal{D}_{\text{sem},0} + 2 G \sqrt{ \rho_{\text{sem},1} \mathcal{D}_{\text{sem},0} } \notag\\
\hspace{-0.3cm}&& = \left( G\sqrt{\rho_{\text{sem},1}} + \sqrt{\mathcal{D}_{\text{sem},0}} \right)^2,
\end{eqnarray}
where (a) follows from the concavity of the square root function (Jensen's inequality); (b) follows from the $G$-Lipschitz continuity of the decoder, i.e., $ \|\bm{\mathcal{J}}_{g}(\bm{r} )\bm{s} \| \leq G \, \|\bm{s}  \|$, and hence $ B \leq G^2 \|\bm{s}  \|^2$; (c) follows from the total attack power constraint $ \|\bm{s}  \|^2 = \rho_{\text{sem},1}$.
From \eqref{e:rho_sem_inequality} we obtain a lower bound on the required attack power.

\begin{lem}\label{lem:bound1_rho_sem}
Consider a semantic communication system under adversarial attack. The minimum attack power needed to achieve an average distortion $\mathcal{D}_{\text{sem},1}$ satisfies
\begin{equation}
\rho_{\text{sem},1} \geq 
\frac{\bigl( \sqrt{\mathcal{D}_{\text{sem},1}}
- \sqrt{\mathcal{D}_{\text{sem},0}} \bigr)^2}{G^2},
\end{equation}
where $\mathcal{D}_{\text{sem},0}$ is the distortion in the absence of an attack.
\end{lem}

To enable a fair comparison across different system architectures, we will later fix the target distortion to a common value, i.e., $\mathcal{D}_{\text{sem},1}=\mathcal{D}^*$. On the other hand, an explicit characterization of $\mathcal{D}_{\text{sem},0}$, which captures the system's intrinsic resilience to channel noise, is necessary to complete the bound and to quantify the intrinsic resilience of semantic communication.

\begin{lem}
\label{lem:D_sem_noise}
For a semantic communication system operating without adversarial perturbations, the average distortion can be expressed as
\begin{eqnarray}\label{e:D_sem_noise}
\mathcal{D}_{\text{sem},0} \hspace{-0.2cm}&=&\hspace{-0.2cm}
\sigma^2_\omega \mathbb{E}_{\bm{z}}\left[  \sum_i  \sigma_{\bm{\mathcal{J}}_{ \! g\!} ,i}^2  \left( \left|h \right|  \bm{z} \right) \right] 
\leq N \sigma^2_\omega G^2,
\end{eqnarray}
where $\{\sigma_{\bm{\mathcal{J}}_{ \! g\!} ,i}(\cdot)\}_i$ denote the singular values of the Jacobian matrix.
\end{lem}

\begin{proof}
See Appendix \ref{App:A}.
\end{proof}

Combining Lemmas~\ref{lem:bound1_rho_sem} and \ref{lem:D_sem_noise} yields an explicit lower bound on the attack power required for a semantic communication system to reach a prescribed target distortion $\mathcal{D}^*$.

\begin{thm}\label{thm:rho_sem_bound}
Consider a semantic communication system subject to adversarial perturbations. The minimum attack power $\rho^*_{\text{sem},1}$ needed to achieve a target distortion $\mathcal{D}^*$ is lower bounded by
\begin{equation}\label{e:sem_bound}
\rho^*_{\text{sem},1}  \geq  
\left( 
\frac{\sqrt{\mathcal{D}^*}}{G }
 -  
\sqrt{  N \sigma_\omega^2}  
\right)^2.
\end{equation}
\end{thm}

Lemma~\ref{lem:D_sem_noise} and Theorem~\ref{thm:rho_sem_bound} together reveal a fundamental connection between the operational design of semantic communication and its resilience to adversarial attacks.

Lemma~\ref{lem:D_sem_noise} shows that, even in the absence of an attacker, the end-to-end distortion $\mathcal{D}_{\text{sem},0}$ is bounded by $N\sigma_\omega^2 G^2$. This implies that to achieve low distortion under channel noise, the decoder must be trained to have a small Lipschitz constant $G$. In other words, the very process of optimizing for noisy channels in DeepJSCC implicitly regularizes the decoder, forcing it to be smooth and insensitive to small input perturbations.

Theorem~\ref{thm:rho_sem_bound} then quantifies how this smoothness translates into adversarial robustness. A smaller $G$ makes the right-hand side of \eqref{e:sem_bound} larger, meaning an adversary must expend more power to achieve the same level of distortion. 
Consequently, the same architectural property that enables efficient communication over noisy channels directly contributes to an inherent robustness against adversarial perturbations.

\begin{rem}
The above insight challenges the prevailing belief that semantic communication, because it relies on DL, is inevitably fragile. On the contrary, the implicit regularization arising from the DeepJSCC training endows the system with a form of built-in resilience. Thus, semantic communication may not be more vulnerable than classical systems; instead, its design for noisy channels naturally promotes a degree of adversarial robustness that has been overlooked in prior work.
\end{rem}

\subsection{Attack Power Analysis for Classical Systems}
Section~\ref{sec:IVA} established a lower bound on the attack power required to achieve a target distortion in semantic communication systems. We now turn our attention to classical SSCC systems, with the goal of deriving an upper bound on the minimum attack power needed to induce a given distortion.

A fundamental distinction between the two paradigms lies in their design philosophies. In semantic communication, the DeepJSCC encoder and decoder are jointly optimized in an end-to-end manner to directly minimize the distortion. In contrast, classical SSCC systems are designed based on a separation principle: channel coding is designed to maximize the reliable transmission rate over a noisy channel, while source coding independently minimizes the source distortion given the rate provided by the channel code.

Directly analyzing the distortion and the corresponding minimum attack power $\rho^*_{\text{sscc},1}$ for a general adversarial perturbation is highly challenging. To make the problem tractable, we consider a specific, yet insightful, attack strategy: the adversary injects an AWGN attack with power $\rho_{\text{sscc},a}$. We denote this specific attack with the subscript ``$a$'' to distinguish it from the general attack denoted by ``$1$''. The key insight is that any attack that is more effective than AWGN would require less power to achieve the same distortion. This leads to the following lemma.

\begin{lem}\label{lem:AWGNattack}
The minimum attack power required to achieve the target distortion $\mathcal{D}^*$ is no greater than that required by an AWGN attack, that is,
\begin{equation}
    \rho^*_{\text{sscc},1} \leq \rho^*_{\text{sscc},a}.
\end{equation}  
\end{lem}

This lemma allows us to upper-bound the SSCC system's vulnerability to the most effective attack by analyzing its response to a simpler, statistically characterized one. Consequently, we can focus our efforts on deriving an expression for $\rho^*_{\text{sscc},a}$, which serves as a conservative bound on the true minimum attack power.

The primary effect of an AWGN attack is to degrade the effective SNR at the receiver. Treating the attack as an additional source of noise, the channel capacity under this combined noise becomes
\begin{equation}\label{e:capacity}
\hspace{-0.2cm} \mathcal{C}  \!=\!   \frac{N}{2}  \log  \left(  1  \!+\!  \frac{|h|^2}{\rho_\omega  \!+\!  \rho^*_{\text{sscc},a}}  \right)  
 \!=\!   \frac{N}{2}  \log  \left(  1  \!+\!   \frac{ \eta N \sigma^2_{ \omega} }{N  \sigma_{ \omega}^2  \!+\!  \rho^*_{\text{sscc},a}}  \right),
\end{equation}
where we have expressed the result in terms of the nominal SNR $\eta$ for a channel without attack.

From an information-theoretic perspective, for a given target distortion $\mathcal{D}$, Shannon's rate-distortion theory dictates the minimum source coding rate $\mathcal{R}(\mathcal{D})$ required to represent the source. For a source $\bm{x} \in \mathbb{R}^M$ with MSE distortion, this rate-distortion function is lower-bounded by
\begin{eqnarray}\label{e:rate_distortion}
&&\hspace{-1.2cm}\mathcal{R}(\mathcal{D}) = \min_{p(\hat{\bm{x}}|\bm{x}):\, \mathbb{E} \left[\|\bm{x}-\hat{\bm{x}}\|^2\right]\le \mathcal{D}} \mathcal{I}(\bm{x};\hat{\bm{x}}) = \mathcal{H}(\bm{x}) - \mathcal{H}(\bm{x}|\hat{\bm{x}}) \notag\\
&&\hspace{-0.8cm}\overset{(a)}{\geq} 
 \mathcal{H}(\bm{x})\! -\! \frac{M}{2} \log{\!\left( 2 \pi e \!\cdot\! \frac{1}{M} \mathcal{D}\! \right)} 
\!\overset{(b)}{=} \frac{M}{2} \log{\left(  \frac{M \Theta_x}{\mathcal{D}}   \right)},
\end{eqnarray}
where $\mathcal{I}$ and $\mathcal{H}$ denote mutual information and differential entropy, respectively.
The inequality in (a) follows from the fact that, for a given MSE, the conditional differential entropy $\mathcal{H}(\bm{x}|\hat{\bm{x}})$ is maximized when the reconstruction error $\bm{x}-\hat{\bm{x}}$ follows a Gaussian distribution. In (b), we introduce the entropy power $\Theta_x$, defined as $\Theta_x \triangleq \frac{1}{2\pi e}\exp \left(\frac{2}{M}\mathcal{H}(\bm{x})\right)$, which captures the fundamental compressibility of the source.

The source-channel separation theorem states that reliable communication with an average distortion $\mathcal{D}^*$ is achievable if and only if the required source coding rate does not exceed the channel capacity, i.e., $\mathcal{R}(\mathcal{D}^*) \leq \mathcal{C}$. Substituting the expressions from \eqref{e:rate_distortion} and \eqref{e:capacity} into this condition yields the fundamental inequality that must be satisfied for the system to operate at distortion $\mathcal{D}^*$ under the AWGN attack:
\begin{equation} \label{e:condition_derivation}
\hspace{-0.3cm} \frac{M}{2} \log{\left(  \frac{M \Theta_x}{\mathcal{D}^*}   \right)} \le \frac{N}{2}  \log  \left(  1  +   \frac{ \eta N \sigma^2_{ \omega} }{N  \sigma_{ \omega}^2  +  \rho^*_{\text{sscc},a}}  \right).
\end{equation}

Rearranging this inequality provides a lower bound on the achievable distortion, given an attack of power $\rho^*_{\text{sscc},a}$:
\begin{equation}\label{e:bound_D_G}
\mathcal{D}^* \ge \frac{M \Theta_x}{\left(  1  +    \eta \cdot \frac{ N  }{N    +  \rho^*_{\text{sscc},a}/\sigma^2_{ \omega}}  \right)^{\frac{N}{M}}} .
\end{equation}

Finally, by inverting this relationship and combining it with Lemma~\ref{lem:AWGNattack}, we obtain an explicit upper bound on the minimum attack power required to force the classical system to achieve a distortion no less than $\mathcal{D}^*$.

\begin{thm}\label{thm:rho_G}
Consider a classical SSCC system operating under an adversarial perturbation. 
The minimum attack power $\rho^*_{\text{sscc},1}$ required to achieve a target distortion $\mathcal{D}^*$ is upper bounded by
\begin{equation}\label{e:sscc_bound}
\rho^*_{\text{sscc},1} \le  \frac{\eta N\sigma_\omega^2} {\left(\frac{M \Theta_x}{\mathcal{D}^*} \right)^{M/N} -1} - N\sigma_\omega^2.
\end{equation}
\end{thm}

The upper bound in \eqref{e:sscc_bound} offers several fundamental insights into the adversarial vulnerability of classical SSCC systems.
\begin{rem}
    The bound can be interpreted as how much additional equivalent noise power must be injected on top of the existing channel noise to push the system distortion to $\mathcal{D}^*$.
\end{rem}

\begin{rem}
    For a fixed channel use $N$ and source dimension $M$, a source with higher entropy power (i.e., one that is less compressible) leads to a larger $\Theta_x$, which increases the numerator $M\Theta_x/\mathcal{D}^*$ and consequently reduces the required attack power. In other words, sources that are inherently harder to compress are also more vulnerable to adversarial perturbations.
\end{rem}

\begin{rem}
    The bound also quantifies the vulnerability amplification effect under high-rate compression. When the system operates with a large bandwidth compression ratio, i.e., when $M/N$ is large, the denominator of the first term decreases exponentially due to the exponent $M/N$. This results in a rapid reduction in the attack power threshold $\rho^*_{\text{sscc},1}$. Thus, pursuing higher efficiency (greater compression) in classical systems exponentially magnifies their adversarial fragility.
\end{rem}

\subsection{Comparative Analysis}

The bounds in Theorems \ref{thm:rho_sem_bound} and \ref{thm:rho_G} lead to a sufficient condition under which the semantic communication system requires higher attack power to reach the same distortion than the classical SSCC system. 

\begin{thm}\label{thm:condition}
In an adversarial environment, the semantic communication system is more robust than the classical SSCC system, i.e., $\rho^*_{\text{sem},1} \geq \rho^*_{\text{sscc},1}$, for achieving the same distortion $\mathcal{D}^*$ if
\begin{eqnarray} \label{e:condition}
\left( 
\frac{\sqrt{\mathcal{D}^*}}{G }
 -  
\sqrt{  N \sigma_\omega^2}  
\right)^2 \geq  \left[\frac{\eta } {\left(\frac{M \Theta_x}{\mathcal{D}^*} \right)^{\frac{M}{N}} -1} - 1 \right]  N \sigma_\omega^2.
\end{eqnarray}
\end{thm}

\begin{figure}
    \centering
    \includegraphics[width=0.75\linewidth]{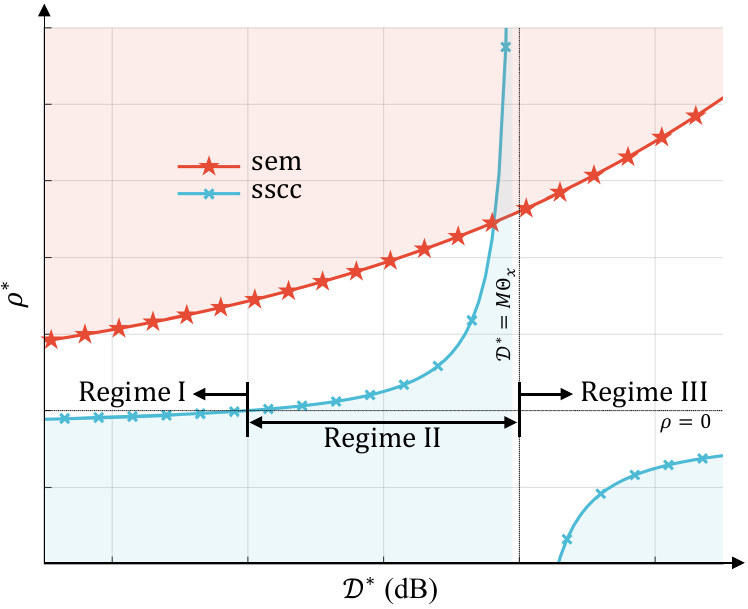}
    \caption{An illustration of the sufficient condition in \eqref{e:condition}.}
    \label{fig:condition}
\end{figure}

To gain insight into when this condition can be satisfied, we analyze the behavior of the left-hand side (LHS) and right-hand side (RHS) as functions of the target distortion $\mathcal{D}^*$, as illustrated in Fig.~\ref{fig:condition}.

The LHS is a squared term, hence always non-negative. For $\mathcal{D}^*$ large enough that $\sqrt{\mathcal{D}^*}/G > \sqrt{N\sigma_\omega^2}$, it grows approximately linearly with $\mathcal{D}^*$; for very small $\mathcal{D}^*$, it may decrease initially but eventually increases with $\mathcal{D}^*$. The RHS, in contrast, has a more intricate structure determined by the source entropy power $\Theta_x$, the bandwidth ratio $N/M$, and the nominal SNR $\eta$. Its behavior can be divided into three distinct regimes based on the value of $\mathcal{D}^*$ relative to the source entropy power $M\Theta_x$.

\begin{itemize}[leftmargin=0.5cm]
\item \textit{Regime I}: $0 < \mathcal{D}^* < M\Theta_x(\eta+1)^{-N/M}$.
Here the RHS is strictly negative, thus the inequality \eqref{e:condition} holds automatically in this regime. We point out that a negative upper bound of $\rho^*_{\text{sscc},1}$ is physically meaningless. It merely reflects that, according to the bound, even zero attack power would already force the classical SSCC system to exceed $\mathcal{D}^*$. In this regime, the classical system is inherently incapable of achieving the target distortion $\mathcal{D}^*$.
\item \textit{Regime II}: $M\Theta_x(\eta+1)^{-N/M} \leq \mathcal{D}^* < M\Theta_x$.
In this interval the RHS is non-negative and increases monotonically with $\mathcal{D}^*$, becoming arbitrarily large as $\mathcal{D}^*$ approaches $M\Theta_x$ from below. The LHS also increases with $\mathcal{D}^*$, but remains finite. Therefore, the inequality \eqref{e:condition} is not automatically satisfied; it holds only when $\mathcal{D}^*$ exceeds a certain threshold where the LHS curve crosses the RHS curve (see Fig.~\ref{fig:condition}). The location of this threshold depends on all system parameters: a smaller Lipschitz constant $G$ raises the LHS, making the condition easier to satisfy for a wider range of $\mathcal{D}^*$. Conversely, a higher SNR $\eta$ or a larger bandwidth ratio $N/M$ steepens the RHS curve near $M\Theta_x$, potentially pushing the threshold to the right.
\item \textit{Regime III}: $\mathcal{D}^* > M\Theta_x$.
Here the denominator $\bigl( \frac{M\Theta_x}{\mathcal{D}^*} \bigr)^{M/N} - 1$ becomes negative, causing the entire RHS to be negative. As in Regime~I, a negative RHS implies that the classical system cannot achieve the target distortion $\mathcal{D}^*$ (since even with no attack the required rate would exceed channel capacity). Hence the comparison again reduces to a scenario where the semantic system is viable while the classical one is not.
\end{itemize}

The above analysis demonstrates that the feasible region characterized by \eqref{e:condition} is nonempty and, in fact, spans a broad range of distortion levels. This observation substantiates that semantic communication are not inherently more fragile than SSCC systems. On the contrary, across a substantial distortion regime, they can achieve superior robustness against adversarial perturbations.

\begin{rem}
    It is worth noting that the bounds employed in Theorems~\ref{thm:rho_sem_bound} and~\ref{thm:rho_G} are conservative. The actual attack power needed to achieve a given distortion may be lower for the classical system (the bound is an upper bound) and higher for the semantic system (the bound is a lower bound). Consequently, the sufficient condition \eqref{e:condition} likely underestimates the true robustness of semantic communication. In practice, with non‑Gaussian error distributions and decoder Jacobians that rarely attain the worst‑case Lipschitz constant, semantic systems may exhibit superior robustness even for distortion values that lie outside the region identified here.
\end{rem}

\section{Attack Methodology}\label{sec:V}
Section~\ref{sec:IV} provided a theoretical comparison of adversarial robustness, suggesting a counterintuitive result that semantic communication systems can exhibit stronger adversarial robustness than classical systems. Yet, the analysis is information-theoretic and relied on bounding arguments. 

To validate these theoretical findings empirically, we must directly compare the minimum attack power required to induce a target distortion in each system. However, existing attack methodologies from the literature are ill-suited for this task. 
For classical systems, prior work on intelligent jamming has largely been confined to information-theoretic analyses under arbitrarily varying channel models, offering no practical attack strategies that exploit the structural vulnerabilities of specific code families like LDPC. Moreover, the discrete, non-differentiable nature of coded modulation and iterative decoding renders standard gradient-based adversarial attacks inapplicable. 
For semantic systems, while DeepJSCC is differentiable, existing attacks, such as FGSM or PGD, are designed for classification tasks and aim to cause misclassification with fixed perturbation budgets, not to find the minimal power required to achieve a target distortion in a regression setting. 

To address these gaps, this section develops two purpose-built attack methodologies for classical and semantic communication systems, respectively. These attack strategies enable a rigorous empirical comparison, allowing us to directly assess the practical robustness of each architecture and provide the first controlled experimental validation of semantic communication's adversarial resilience.

\subsection{Vulnerable Set Attack for Classical Systems}\label{sec:VA}
The robustness of classical communication systems fundamentally relies on error correction codes. Decades of coding theory have focused almost exclusively on protecting against random noise, such as AWGN, while the behavior of channel codes under intelligently crafted, targeted perturbations remains almost entirely unexplored. This section fills this gap by designing attacks that exploit both the structural vulnerabilities of classical codes and the instantaneous channel state to induce decoding failure with minimal energy. In particular, we focus on the LDPC code, a widely-used channel code renowned for its near-capacity performance under random noise.

An LDPC code is defined by a sparse parity-check matrix $\bm{H} \in \{0,1\}^{m \times n}$. It is conveniently represented as a Tanner graph $\mathcal{G} = (\{v\}, \{c\}, \{\varepsilon\})$, where $\{v\}$, $\{c\}$, and $\{\varepsilon\}$ denote the sets of variable nodes (VNs, corresponding to coded bits), check nodes (CNs, corresponding to parity checks), and edges, respectively \cite{gallager2003low}. An edge between $v_i$ and $c_j$ exists if and only if $\bm{H}_{ij}=1$.
Decoding is performed via the belief propagation (BP) algorithm \cite{ryan2009channel,richardson2002design}. The decoder first computes a log-likelihood ratio (LLR) $L_i^{(0)}$ for each VN $v_i$ from the received signal. It then iteratively exchanges messages along the edges of the Tanner graph. If at any iteration the hard decisions based on the current LLRs satisfy all parity checks, decoding stops successfully; otherwise, it continues until a maximum number of iterations is reached, at which point a decoding failure is declared.
In the following, we will exploit the structured vulnerabilities of such an iterative, graph-based structure to design targeted perturbations. 

\subsubsection{Vulnerable set}
To induce decoding failure with minimal energy, an adversary should concentrate its power on the VNs most likely to be responsible for such failures. This raises a central question: how can we identify, for a given LDPC code, which nodes are the most attractive targets?

The coding theory literature offers extensive insights into why BP decoding fails. Common culprits include short cycles in the Tanner graph, particularly 4-cycles, which introduce correlations between messages and hinder convergence; high-degree nodes, whose errors can propagate widely; and local graph structures that are only weakly constrained by the parity-check equations. These observations suggest that vulnerability is not uniformly distributed across nodes; rather, certain nodes are structurally predisposed to play a role in decoding failures.

Inspired by these insights, we propose to pre-identify a \emph{vulnerable set} $\mathcal{T} \subseteq \mathcal{V}$. This set comprises VNs that, based solely on the static code structure, are deemed most likely to be implicated in decoding failures. The adversary does not restrict its attack exclusively to these nodes; rather, it uses this set to guide a weighted allocation of the attack budget. By assigning higher attack weights to nodes in $\mathcal{T}$, the adversary can focus its energy where it is structurally most effective. Importantly, because $\mathcal{T}$ depends only on the parity-check matrix $\bm{H}$, it can be computed offline once per code and reused for all received signals. The online attack then dynamically allocates power across all nodes, but with a bias towards those in the vulnerable set, based on real-time LLRs. 

\subsubsection{Identifying the vulnerable set}
To construct the vulnerable set, we assign each VN $v_i$ a static vulnerability score $\varphi_i$ that combines four complementary features. Each feature is motivated by known causes of BP decoding difficulty and is computed exclusively from $\bm{H}$.

\textit{Node degree}.
A simple but essential indicator is the VN's degree $|\mathcal{N}_c(i)|$. High-degree nodes participate in many parity checks, so errors there have multiple opportunities to disrupt parity satisfactions. We include degree as a baseline structural feature:
\begin{equation}\label{eq:degree_feature}
\varphi_{\mathrm{deg}}(i) = |\mathcal{N}_c(i)|.
\end{equation}

\textit{Two-hop reachability}.
A node's influence extends beyond its immediate neighbors. If a node can affect many others within two hops (through a shared check node), corrupting it may have widespread consequences. Let $\mathcal{N}_c(i)$ be the set of CNs adjacent to VN $i$, and for a CN $c$, let $\mathcal{N}_v(c)$ be its adjacent VNs. The two-hop reachable set for node $i$ is
\begin{equation}
 \mathcal{R}_2(i)=\Bigl(\bigcup_{c\in \mathcal{N}_c(i)} \mathcal{N}_v(c)\Bigr)\setminus \{i\},   
\end{equation}
and the reachability feature is its cardinality:
\begin{equation}\label{eq:2hop_feature}
\varphi_{\mathrm{hop}}(i) = |\mathcal{R}_2(i)|.
\end{equation}
Nodes with large two-hop reach can propagate errors widely, making them attractive targets.

\textit{Four-cycle participation}.
Short cycles, especially 4-cycles, degrade BP performance by creating early correlations that mislead message updates. As shown in Fig.~\ref{fig:LDPC}, a VN that participates in many 4-cycles is embedded in a local structure that hinders convergence, making it more likely to be involved in a decoding failure. For VNs $i$ and $j$, let $K_{ij}$ be the number of check nodes they share. Then the number of 4-cycles involving the pair $(i,j)$ is $\binom{K_{ij}}{2}$. Summing over all $j$ gives the total 4-cycle involvement of node $i$:
\begin{equation}\label{eq:4cycle_feature}
\varphi_{\mathrm{cyc}}(i) = \sum_{j=1}^{n} \binom{K_{ij}}{2}.
\end{equation}

\begin{figure}[t!]
\centering
\includegraphics[width=0.95\columnwidth]{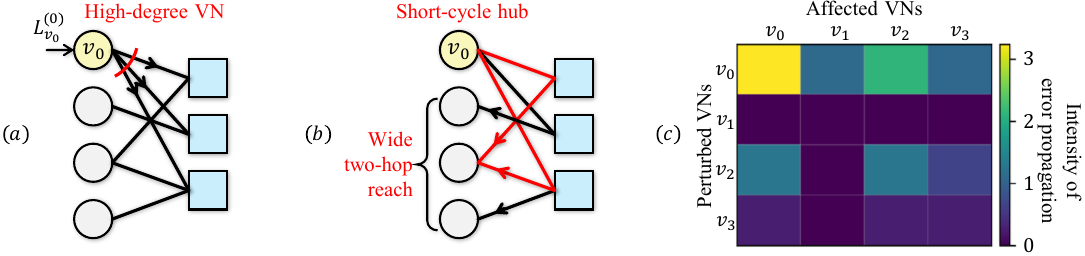}
\vspace{-0cm}
\caption{An illustration of the structural vulnerabilities of LDPC. The VN $v_0$ (a) has a high degree, and (b) is a short-cycle hub and has wide two-hop reach. The error propagation heatmap is shown in (c), where color intensity reflects how strongly a perturbation injected at each VN propagates to others. The resulting vulnerability order is $v_0>v_2>v_3>v_1$.
}
\label{fig:LDPC}
\end{figure}

\textit{Weak constraint participation}.
Classical trapping-set analysis has shown that certain linear combinations of VNs are only weakly constrained by the parity-check equations; these directions correspond to small singular values of $\bm{H}$ when viewed as a real-valued matrix. Although $\bm{H}$ is binary, its singular value decomposition (SVD) $\bm{H} = \bm{U}\bm{\Sigma}\bm{V}^{\top}$ reveals the real-valued subspaces where the parity checks offer little constraint. In particular, a small singular value $\sigma$ indicates that there exists a direction in the VN space (a right singular vector) that is mapped to a very small vector by $\bm{H}$. VNs with large coefficients in that direction are therefore part of a weakly constrained combination, making them structurally vulnerable.

Concretely, let $\sigma_1 \ge \sigma_2 \ge \cdots \ge \sigma_r > 0$ be the nonzero singular values, and let $\bm{V}$ be the corresponding right singular vectors. We focus on the $k$ smallest singular values, where $k$ is chosen at the elbow of the singular-value curve. For each VN $j$, we define
\begin{equation}\label{eq:svd_feature}
\varphi_{\mathrm{svd}}(j) = \sum_{i=1}^{k} \frac{|\bm{V}_{j,r-i+1}|}{\sigma_{r-i+1} + \varepsilon},
\end{equation}
where $\varepsilon>0$ prevents division by zero. A high score indicates strong participation in weakly constrained directions.

All four features are computed from $\bm{H}$ alone, enabling offline precomputation. Once computed, each feature is normalized to $[0,1]$ via min-max scaling. For example, $\varphi_{\mathrm{svd}}(i)$ is normalized as
\begin{equation}
\widetilde{\varphi}_{\mathrm{svd}}(i)=\frac{\varphi_{\mathrm{svd}}(i)-\min_j \varphi_{\mathrm{svd}}(j)}{\max_j \varphi_{\mathrm{svd}}(j)-\min_j \varphi_{\mathrm{svd}}(j)}.
\end{equation}
Similarly, we can obtain the normalized $\widetilde{\varphi}_{\text{deg}}(i)$, $\widetilde{\varphi}_{\text{cyc}}(i)$, and $\widetilde{\varphi}_{\text{hop}}(i)$, respectively.

The total vulnerability score is obtained through a weighted linear combination:
\begin{equation}\label{eq:vulnerability_fusion}
\varphi_i=w_1\widetilde{\varphi}_{\mathrm{deg}}(i)+w_2\widetilde{\varphi}_{\mathrm{hop}}(i)+w_3\widetilde{\varphi}_{\text{cyc}}(i)+w_4\widetilde{\varphi}_{\text{svd}}(i),
\end{equation}
where $w_1\sim w_4$ are weighting coefficients.

With the vulnerability scores $\{\varphi_i\}_{i=1}^n$ computed, we now determine the vulnerable set $\mathcal{T}$ by thresholding. A fixed threshold would be suboptimal, as the distribution of scores can vary across different codes and channel conditions. Instead, we employ an adaptive thresholding scheme that adjusts to the statistical properties of the scores.

Let $\mu^{(\varphi)}$ and $\sigma^{(\varphi)}$ denote the sample mean and standard deviation of $\{\varphi_i\}_{i=1}^n$. We first compute an initial threshold based on the spread of the distribution:
\begin{equation}
\tau_{\text{initial}}=
\begin{cases}
\mu^{(\varphi)}+2\sigma^{(\varphi)}, & \text{if }\sigma^{(\varphi)}<0.1,\\[4pt]
P^{(\varphi)}_{85}, & \text{otherwise},
\end{cases}
\end{equation}
where $P^{(\varphi)}_x$ represents the $x$-th percentile. This rule adapts to the score concentration: when scores are tightly clustered (small $\sigma^{(\varphi)}$), a deviation-based threshold isolates outliers; when scores are more spread, a percentile threshold ensures a stable fraction of nodes is selected.

To prevent extreme values, we clip the initial threshold to lie between the 70th and 95th percentiles:
\begin{equation}
\tau = \min\Bigl(\max\bigl(\tau_{\text{initial}}, P^{(\varphi)}_{70}\bigr),\, P^{(\varphi)}_{95}\Bigr).
\end{equation}
This dual-stage selection ensures that the threshold is neither too aggressive (selecting too few nodes) nor too lenient (selecting too many), while remaining adaptive to the score distribution.

The final vulnerable set is then defined as the nodes whose scores meet or exceed this threshold:
\begin{equation}\label{eq:VS}
\mathcal{T}=\bigl\{\,i\mid \varphi_i\ge \tau\,\bigr\}.
\end{equation}
This set $\mathcal{T}$ is computed once per code and stored for online use.

\subsubsection{Vulnerable set attack}
With the vulnerable set $\mathcal{T}$ fixed offline, the online attack must decide how much power to allocate to each node based on the instantaneous channel state. The key insight is that even within the vulnerable set, nodes differ in their current reliability as reflected by the received LLRs. A node with a very small LLR magnitude is already close to being unreliable; perturbing it further requires little energy. Conversely, a node with a large LLR magnitude is confidently decoded and requires a larger perturbation to flip.

We therefore design a gradient-driven attack that allocates power proportionally to each node's vulnerability, with nodes in $\mathcal{T}$ receiving a higher allocation. The attack proceeds in two steps: first, compute the direction that maximally reduces decoding reliability for each symbol; second, scale the perturbation based on vulnerable-set membership.

For a received complex symbol $r_i = r^{(x)}_i + jr^{(y)}_i$, let 
\begin{equation*}
    f(r_i) = \sum_{b=1}^{B} \left|L_b\left(r_i, \frac{\sigma^2_\omega}{|h|^2}\right)\right|
\end{equation*}
be the sum of absolute LLRs over the $B$ bits mapped to that symbol (e.g., $B=2$ for QPSK, $B=4$ for 16-QAM). This quantity serves as a proxy for the decoder's confidence in that symbol; a smaller $f(r_i)$ indicates that the symbol's bits are closer to the decision boundary and thus more likely to be flipped by the BP decoder. The adversary's goal is to minimize the total confidence $\sum_i f(r_i)$ to induce a decoding failure.

To achieve this minimization with minimal perturbation power, we employ a gradient descent approach. Because $f$ is not differentiable at points where any $|L_b|$ is zero, we use a symmetric finite-difference approximation to compute its gradient. The partial derivative with respect to the in-phase component $r^{(x)}_i$ is
\begin{equation}
    \frac{\partial f}{\partial r^{(x)}_i} \approx \frac{f(r^{(x)}_i+\epsilon)-f(r^{(x)}_i-\epsilon)}{2\epsilon},
\end{equation}
and similarly for the quadrature component $r^{(y)}_i$. The step size $\epsilon = 10^{-6} \cdot \max(1, |r_i|)$ scales with the symbol magnitude to maintain numerical accuracy across varying amplitudes. The complex gradient is then formed as
\begin{equation}\label{eq:gradient}
\nabla_i = \frac{\partial f}{\partial r^{(x)}_i} + j\frac{\partial f}{\partial r^{(y)}_i},
\end{equation}
which points in the direction of the steepest ascent of $f$. Since our objective is to minimize $f$, we will move each symbol in the opposite direction, $-\nabla_i$, to maximally reduce the reliability of its bit-wise LLRs.

To concentrate power on the most effective targets, we adopt an adaptive attack strategy that assigns each symbol a weight reflecting both its structural vulnerability (membership in $\mathcal{T}$) and its current sensitivity (inverse of the gradient magnitude). Unlike traditional attack methods, our approach adjusts the attack direction based on the channel state, rather than employing a fixed and rigid attack strategy. Specifically, the one-step perturbation is designed as
\begin{equation}\label{eq:perturbation}
s_{i,t} = \eta \cdot\bigl(1 + e \cdot \mathbb{I}_{{i \in \mathcal{T}}}\bigr) \cdot \frac{-\nabla_i}{|\nabla_i| + \varepsilon},
\end{equation}
where $\varepsilon = 10^{-6}$ prevents division by zero, and $\eta$ is a small step size controlling the attack power added per iteration. The indicator function $\mathbb{I}_{{i \in \mathcal{T}}}$ assigns nodes in the vulnerable set $e$ times the baseline weight, focusing energy on structurally weak points. Critically, the gradient $\nabla_{i,t}$ is recomputed at each step $t$ based on the currently perturbed signal $r_{i,t}$, allowing the attack to adapt to the changing landscape of $f$. The term $1/(|\nabla_i|+\varepsilon)$ acts as an adaptive per-symbol gain: symbols with smaller gradients (i.e., already near a decision boundary) receive a larger relative update, as less energy is needed to push them over the edge.

Recall that our goal is to find the smallest $\rho^*$ such that the attack forces the distortion above $\mathcal{D}^*$. The total perturbation $s_i$ for symbol $r_i$ is obtained by accumulating these one-step updates $s_{i,t}$ iteratively. Starting from a small initial power $\rho_0$, we construct the perturbation using \eqref{eq:perturbation} with the current $\rho$, apply it, and check whether the decoder fails. At each step $t$, the gradients $\nabla_{i,t}$ are recomputed based on the currently perturbed signal $r_{i,t} = r_{i,0} + \sum_{k=1}^{t-1} s_{i,k}$, where $r_{i,0}$ is the initial received signal, which allows the algorithm to adapt to the changing $f(r_{i,t})$. After a series of iterative attacks, the reconstruction process ultimately fails, and we accumulate all the $s_{i,t}$ to determine the optimal $\rho^*$.

This adaptive, iterative scheme directly mirrors the optimization problem in \eqref{eq:problem} by seeking the minimal perturbation path to failure. It avoids over-perturbation by stopping as soon as the target (decoding failure) is achieved, naturally handles varying channel conditions by recalculating gradients, and continuously refines the attacking direction based on the current state. The resulting $\rho^*$ provides an empirical measure of the classical system's robustness, which can be directly compared with the semantic system's counterpart.

\subsection{A RL-based GMS Attack for Classical Systems}
The vulnerable set attack introduced in Section \ref{sec:VA} leverages structural insights from coding theory to craft deterministic, highly efficient perturbations. To validate its effectiveness and establish that its performance is not merely an artifact of a particular heuristic, we develop a complementary, data-driven attack based on RL. This attack, termed the Gaussian mixture sequential (GMS) attack, serves as a powerful and general baseline. Unlike the deterministic vulnerable set attack, GMS attack sequentially adds Gaussian noise to contiguous blocks of symbols. The cumulative effect yields a total perturbation whose distribution is a mixture of Gaussians, offering a flexible and expressive attack strategy that learns to exploit system vulnerabilities directly from interaction.

We model the attack process as a finite-horizon Markov Decision Process (MDP). At each discrete time step $t$, the adversary observes the current state of the communication and chooses an action that determines the location of a stochastic perturbation. The goal is to learn a policy that minimizes the cumulative expected attack power required to force the reconstruction distortion $\mathcal{D}$ above the target threshold $\mathcal{D}^*$. 

\subsubsection{MDP}
The MDP components are defined below.

\textit{State}.
The state at step $t$ is the perturbed signal vector, which we denote as $\widetilde{\bm{y}}^{(t)} \in \mathbb{R}^N$. We initialize $\widetilde{\bm{y}}^{(0)} = \bm{r}$, which is the clean received signal before any adversarial intervention. The state thus captures the full instantaneous condition of the channel output after all previous attack steps. This high-dimensional, continuous-valued representation necessitates function approximation for value estimation, which we address using a DNN.

\textit{Action}. 
At each step $t$, the action $a^{(t)} \in \{1, 2, \ldots, N\}$ selects a starting index for the attack. Once the starting position is chosen, the adversary attacks a contiguous block of $\lceil \log_2 N \rceil$ consecutive symbols, beginning at $a^{(t)}$ and wrapping around modulo $N$ if necessary. 

For the selected block, the adversary injects zero-mean Gaussian noise whose variance is matched to the instantaneous power of the current signal in that region. Specifically, for each symbol index $j$ in the attack block, we generate i.i.d. Gaussian noise samples $n_j^{(t)} \sim \mathcal{N}(0, \sigma_t^2)$, where the variance $\sigma_t^2$ is set to the average power of the current signal $\widetilde{\bm{y}}^{(t-1)}$ over the attack region:
\begin{equation*}
\sigma_t^2 = \frac{1}{B} \sum_{k=0}^{B-1} \left| \widetilde{\bm{y}}_{a^{(t)}+k}^{(t-1)} \right|^2.
\end{equation*}
This adaptive scaling ensures that the perturbation magnitude remains commensurate with the local signal strength, preventing the agent from wasting power on already-weak symbols.

The signal is then updated via a convex combination that interpolates between the current signal and the noise:
\begin{equation}\label{eq:gms_update}
\widetilde{\bm{y}}_j^{(t)} = (1-\alpha) \widetilde{\bm{y}}_j^{(t-1)} + \alpha \bm{n}_j^{(t)}, \quad \forall j \in \mathcal{B}_t,
\end{equation}
where $\mathcal{B}_t = \{a^{(t)}, a^{(t)}+1, \ldots, a^{(t)}+B-1\}$ (with indices taken modulo $N$), and $\alpha \in [0,1]$ is a fixed mixing coefficient that controls the noise injection strength. For symbols outside $\mathcal{B}_t$, the signal remains unchanged. The mixing parameter $\alpha$ is a hyperparameter that governs the trade-off between attack aggressiveness and step-wise control; a smaller $\alpha$ allows finer-grained adjustments but may require more steps to achieve the target distortion.

The sequential application of these localized perturbations results in a cumulative attack distribution that is a mixture of Gaussians, hence the name GMS attack.


\textit{Reward.}
The ultimate objective is to minimize the expected total attack power $\rho^* = \mathbb{E}[|\bm{s}|^2]$, where $\bm{s} = \widetilde{\bm{y}}^{(T)} - \bm{r}$ is the cumulative perturbation accumulated over an episode of length $T$. To cast this as a reward maximization problem, we define an immediate reward that penalizes each increment of attack power.

Let $\bm{s}^{(t)} = \widetilde{\bm{y}}^{(t)} - \bm{r}$ denote the cumulative perturbation after $t$ steps, with $\bm{s}^{(0)} = \bm{0}$. We define the immediate reward as the negative of the increase in total attack power from step $t$ to $t+1$. That is
\begin{equation}
\mathfrak{R}^{(t)} = |\bm{s}^{(t)}|^2 - |\bm{s}^{(t+1)}|^2.
\end{equation}
With this definition, the cumulative reward over an episode of $T$ steps telescopes to $-|\bm{s}^{(T)}|^2$. Thus, maximizing the expected cumulative reward is exactly equivalent to minimizing the expected total attack power $\rho^* = \mathbb{E}[|\bm{s}^{(T)}|^2]$.

The episode terminates at step $T$ if the distortion after applying the $T$-th action meets or exceeds the target, i.e., $D(\bm{x}, \widehat{\bm{x}}^{(T)}) \ge \mathcal{D}^*$, where $\widehat{\bm{x}}^{(T)}$ is the reconstruction from $\widetilde{\bm{y}}^{(T)}$. No reward is assigned beyond termination.

\subsubsection{Training and attack generation}
To solve the MDP, we employ the Deep Q-Network (DQN) algorithm, which is well-suited for problems with high-dimensional state spaces and discrete actions.

The training episodes are generated as follows.
Starting from a clean received signal $\bm{r}$, we initialize $\widetilde{\bm{y}}^{(0)} = \bm{r}$.
At each step $t$, the agent selects a starting index $a^{(t)}$ according to the current policy. 
We then compute the adaptive noise variance $\sigma_t^2$ from the current signal in the chosen block, sample Gaussian perturbations $\bm{n}_j^{(t)} \sim \mathcal{N}(0, \sigma_t^2)$ for all $j \in \mathcal{B}_t$, and apply the update in \eqref{eq:gms_update} to obtain $\widetilde{\bm{y}}^{(t+1)}$.
The received signal is then passed through the receiver chain (demodulator, channel decoder, source decoder) to obtain the reconstruction $\widehat{\bm{x}}^{(t+1)}$ and compute the distortion $D(\bm{x}, \widehat{\bm{x}}^{(t+1)})$. The episode continues until $\mathcal{D} \geq \mathcal{D}^*$.

After training converges, we use the learned policy to generate attacks for evaluation. For a given test sample, we run the trained agent to produce a sequence of actions and corresponding perturbations. The total perturbation is accumulated as $\bm{s} = \widetilde{\bm{y}}^{(T)} - \bm{r}$, and the expected attack power is $\rho^* = \mathbb{E}[|\bm{s}|^2]$.
Because the increments $\Delta \bm{s}^{(t)}$ are not independent (they depend on the state evolution), we estimate $\rho^*$ empirically by averaging over multiple test episodes.

\begin{rem}
The GMS attack provides a powerful and general baseline for several reasons. First, by learning directly from interaction, it can discover vulnerabilities that are not apparent from static structural analysis. Second, the sequential nature allows the attack to adapt its strategy based on intermediate decoding outcomes, potentially achieving the target distortion with lower total power than any fixed one-shot attack. Third, because the perturbation at each step is stochastic, the learned policy is robust to variations in channel noise and source realizations. Comparing the performance of the GMS attack with the deterministic vulnerable set attack will reveal whether the structural insights from Section~\ref{sec:VA} capture the essential vulnerabilities of classical systems, or whether more sophisticated, adaptive strategies can achieve even greater efficiency.
\end{rem}

\subsection{PGA Attack for Semantic Communication}
Having developed structure-aware and learning-based attacks for classical systems, we now turn to the semantic communication counterpart. Unlike classical systems where the adversary must exploit discrete combinatorial structures (e.g., parity-check matrices), semantic communication systems employ continuously differentiable DNN encoders and decoders. This end-to-end differentiability opens a principled attack design: the adversary can leverage gradient information to systematically increase the reconstruction distortion.

We propose the progressive gradient ascent (PGA) attack, an iterative method that finds the minimum-power perturbation $\bm{s}$ such that the reconstruction distortion exceeds $\mathcal{D}^*$. The attack operates under the same white-box assumption as before: the adversary knows the DeepJSCC encoder $f(\cdot;\bm{W}_f)$ and decoder $g(\cdot;\bm{W}_g)$, and observes the received signal $\bm{r}$.

\subsubsection{Attack formulation}
Let $\widehat{\bm{x}}(\bm{r}) = g(\bm{r};\bm{W}_g)$ denote the reconstruction from the clean received signal. When a perturbation $\bm{s}$ is injected, the received signal becomes $\bm{y} = \bm{r} + \bm{s}$, yielding reconstruction $\widehat{\bm{x}}(\bm{y})$. Because the distortion $D\bigl(\bm{x}, \widehat{\bm{x}}(\bm{y})\bigr)$ is a differentiable function of $\bm{y}$ (through the decoder), we can compute its gradient and use it to guide the search for a minimal perturbation.

The proposed PGA attack builds upon the principled framework of projected gradient descent (PGD) \cite{madry2017towards}, which has been established as a universal first-order adversary in the adversarial machine learning literature. Specifically, PGD-based attacks have been shown to provide a tight lower bound on the required perturbation magnitude for a given misclassification target \cite{carlini2017towards,madry2017towards}. By adapting this method to our communication setting, we inherit its theoretical guarantees as a strong adversary while introducing several improvements tailored to our problem.

The PGA attack iteratively builds a perturbation by taking steps in the direction that most rapidly increases the distortion. 
At each iteration $t$, we compute the gradient of the distortion with respect to the current signal.
For the MSE distortion, this gradient can be written as
\begin{equation}\label{e:36}
    \nabla_{\bm{y}} D\bigl(\bm{x}, \widehat{\bm{x}}(\bm{y}^{(t)})\bigr) = 2 \left(\bm{\mathcal{J}}_{g}(\bm{y}^{(t)}) \right)^{\top} \bigl(\bm{x} - \widehat{\bm{x}}(\bm{y}^{(t)})\bigr).
\end{equation}
Then, we design the update direction to be the normalized gradient, which points in the direction of steepest ascent of the distortion. The signal is then updated as
\begin{equation}\label{e:recursion}
    \bm{y}^{(t+1)} \! = \! \bm{y}^{(t)} \! + \! \alpha \! \cdot \! \frac{\nabla_{\bm{y}} D\bigl(\bm{x}, \widehat{\bm{x}}(\bm{y}^{(t)})\bigr)}{\|\nabla_{\bm{y}} D\bigl(\bm{x}, \widehat{\bm{x}}(\bm{y}^{(t)})\bigr)\| \! + \! \varepsilon} 
    \! \triangleq \!  \bm{y}^{(t)} \! + \!  \bm{s}^{(t)},
\end{equation}
with $\alpha$ is the step size and $\varepsilon$ is a small constant to avoid division by zero. Typically, the perturbations are small. Therefore, we stop the update process at the first iteration $T$ such that $D(\bm{x}, \widehat{\bm{x}}(\bm{y}^{(T)})) \ge \mathcal{D}^*$.

\subsubsection{Gradient properties and attack dynamics}
Combining the definitions in \eqref{e:recursion}, the general expression for each attack step $\bm{s}^{(t)}$ can be obtained via recursion and mathematical induction as
\begin{eqnarray}\label{e:s_t}
&&\hspace{-0.8cm} \bm{s}^{(t)} = \alpha^{(t)} \left( \bm{\mathcal{J}}_{\!g}(\bm{y}^{(t)}) \right)^\top \\
&&\hspace{-0.5cm} \left[ \prod_{i=0}^{t-1} \! \left( \! \bm{I}_M \! + \! \alpha^{(i)} \bm{\mathcal{J}}_{\!g}(\bm{y}^{(i)}) \! \left( \! \bm{\mathcal{J}}_{\!g}(\bm{y}^{(i)}) \! \right)^{\!\!\top} \! \right) \!\! \right] \!\! \left( \! g(\bm{y}^{(1)}; \! \bm{W}_g) \! - \! \bm{x} \! \right), \notag
\end{eqnarray}
where $\bm{\mathcal{J}}_{\!g}(\bm{y}^{(0)}) = \bm{0}$, and $\alpha^{(t)} = \frac{\alpha }{\|\nabla_{\bm{y}} D\bigl(\bm{x}, \widehat{\bm{x}}(\bm{y}^{(t)})\bigr)\| + \varepsilon}$ is a normalized step size. The detailed derivation of \eqref{e:s_t} can be found in Appendix \ref{App:B}.

Analyzing \eqref{e:s_t} yields the following insights:
\begin{itemize}[leftmargin=0.5cm]
\item The rightmost term $(g(\bm{y}^{(1)})-\bm{x})$ represents the initial error driving force, namely the discrepancy between the decoded output $g(\bm{y}^{(1)})$ under channel noise and the ideal target $\bm{x}$. Thanks to noisy training, $(g(\bm{y}^{(1)})-\bm{x})$ would be very small.

\item The middle product term contains matrices of the form $\bm{I}_{M}  +  \alpha^{(i)} \bm{\mathcal{J}}_{g}(\bm{y}^{(i)}) \left( \bm{\mathcal{J}}_{g}(\bm{y}^{(i)})  \right)^{\top}$, which are positive definite with all eigenvalues larger than one. Through successive spectral amplification, this multiplicative structure automatically identifies and strengthens the feature dimensions that are most sensitive to decoding errors. As a result, the attack energy is adaptively concentrated onto the most vulnerable subspace.

\item Finally, the factor $(\bm{J}_g(\bm{y}^{(t)}))^\top$ projects the exponentially amplified error back into the input signal space, and the coefficient $\alpha^{(t)}$ constrains the effective step size to $\alpha$, mitigating the risk of gradient explosion or vanishing and stabilizing the attack dynamics.
\end{itemize}

After $T$ steps, the accumulated perturbation is $\bm{s} = \sum_{t=1}^{T} \bm{s}^{(t)} = \bm{y}^{(T)} - \bm{r}$, and the corresponding attack power is $\rho^* = \|\bm{s}\|^2$.

\subsubsection{C\&W Attack as a Baseline}
To assess the effectiveness of the proposed PGA attack, we compare it against a strong gradient-based baseline: the Carlini-Wagner (C\&W) attack \cite{carlini2017towards}. Originally developed for image classification, the C\&W attack minimizes a weighted sum of the perturbation norm and an adversarial loss that encourages misclassification. Its formulation provides a principled way to find minimal-norm perturbations and has become a standard benchmark in the adversarial machine learning literature. 

Adapting C\&W to the regression task of source reconstruction requires redefining the adversarial objective. Instead of pushing a classifier's output across a decision boundary, the attacker now aims to drive the reconstruction distortion above a target threshold $\mathcal{D}^*$ while minimizing the perturbation power. We therefore formulate the attack as
\begin{equation}\label{e:CW}
\min_{\boldsymbol{s}} \|\boldsymbol{s}\|_2^2 + c \cdot \max\left(D(\hat{\bm{x}}, \bm{x}) - \mathcal{D}^*, -\kappa\right)
\end{equation}
where $c > 0$ is a regularization parameter that trades off the perturbation magnitude against the distortion objective, and $\kappa \geq 0$ is a confidence margin that ensures the attack exceeds the target distortion. To solve the attack objective in \eqref{e:CW}, we employ an iterative optimization procedure using the Adam optimizer.

\section{Numerical Experiments}\label{sec:VI}
In this section, we empirically evaluate the adversarial robustness of both classical and semantic communication systems using the attack strategies developed in Section~\ref{sec:V}. Our goal is to quantify the minimum attack power $\rho^*$ to provide a direct comparison of their inherent resilience under worst-case adversarial interference. We consider two representative tasks: image transmission and massive MIMO CSI transmission.

\subsection{Image Transmission}
We first consider the task of transmitting images over a wireless channel in the presence of an adversary. The CIFAR-10 dataset \cite{krizhevsky2009learning} is used, consisting of $60,000$ RGB images of size $32\times 32\times 3$. We partition the dataset into a training set of $50,000$ images for training the DeepJSCC model, and a test set of $10,000$ images for evaluating both classical and semantic communication systems.

\subsubsection{System setup}
The quality of the reconstructed image is quantified by the peak signal-to-noise ratio (PSNR), defined as
\begin{equation}\label{eq:psNR}
\text{PSNR} = 10 \cdot \log_{10} \left( \frac{\text{Peak}^2}{\text{MSE}} \right) \quad \text{(dB)},
\end{equation}
where $\text{Peak}$ is the maximum possible pixel value and $\text{MSE}$ is the MSE between the original and reconstructed images. For images normalized to $[0,1]$, $\text{Peak}=1$. In our experiments we set a target PSNR of $15$dB, which corresponds to a target MSE of $\mathcal{D}^* \approx 3\times{10^{-2}}$.

To ensure a fair comparison, both classical and semantic systems operate under the same bandwidth and power constraints. We define the \textit{bandwidth ratio} as the number of channel uses $N$ divided by the source dimension. For a CIFAR-10 image, the source dimension is $32 \times 32 \times 3 = 3072$. Throughout this subsection, we fix the bandwidth ratio to $1/4$, i.e., $N = 3072/4 = 768$ channel uses. The transmit power is normalized so that $\mathbb{E}[\|\bm{z}\|^2] = N$.

\textit{Classical system:} The classical transmitter employs SSCC. For source coding we use Better Portable Graphics (BPG), a state-of-the-art image compression standard. The compressed bitstream is then protected by a regular LDPC code, with code rate $R_c$ chosen to match the bandwidth ratio. The coded bits are modulated using either QPSK or 16-QAM with Gray mapping.

\begin{table}[!t]
\centering
\caption{DeepJSCC encoder and decoder architectures and training hyperparameters, where $C$ is the number of image channels, $H$ and $W$ are the spatial dimensions, and $C_{\text{num}}$ is a channel multiplier that determines the final number of channel uses.}
\label{tab:imageDNN}
\scriptsize
\begin{tabular}{lll}
\toprule
\multicolumn{1}{l}{\textbf{Component}} & \textbf{Layer Description} & \textbf{Output Size} \\
\midrule
\multirow{4}{*}{\textbf{Encoder}}
& Conv1 + ReLU & $(16, H/2, W/2)$ \\
& Residual Block & $(16, H/2, W/2)$ \\
& Conv2 + ReLU & $(2C_{\text{num}}, H/4, W/4)$ \\
& Residual Block & $(2C_{\text{num}}, H/4, W/4)$ \\
\midrule
\multirow{6}{*}{\textbf{Decoder}}
& Deconv1 + ReLU & $(32, H/4, W/4)$ \\
& Deconv2 + Sigmoid & $(32, H/4, W/4)$ \\
& Residual Block & $(32, H/4, W/4)$ \\
& Deconv3 & $(16, H/2, W/2)$ \\
& Residual Block & $(16, H/2, W/2)$ \\
& Deconv4 + Sigmoid & $(C, H, W)$ \\
\midrule
\multirow{4}{*}{\textbf{Training}}
& Number of training epochs & 1000 \\
& Batch size & 512 \\
& Learning rate & $10^{-3}$ \\
& Optimizer & AdamW \\
\bottomrule
\end{tabular}
\end{table}

\textit{Semantic system:} The semantic system is implemented as a DeepJSCC architecture. The encoder $f(\cdot;\bm{W}_f)$ is a convolutional neural network (CNN) with residual blocks and ReLU activations, which maps the input image directly to $N$ channel symbols $\bm{z}$. The decoder $g(\cdot;\bm{W}_g)$ consists of deconvolution layers with residual blocks and a final sigmoid activation to ensure the output is in the valid pixel range $[0,1]$. The network is trained end-to-end on the CIFAR-10 training set using the AdamW optimizer with a learning rate of $10^{-3}$ and a batch size of $512$. The loss function is the MSE between the original and reconstructed images. Training proceeds for $1000$ epochs. Detailed DNN layer specifications and hyperparameters are provided in Table~\ref{tab:imageDNN}.

\subsubsection{Attack setup}
The detailed parameter settings for the four attacks introduced in Section~\ref{sec:V} are as follows. All attacks are evaluated over the test set.

\textit{Vulnerable set attack (classical):} The vulnerable set $\mathcal{T}$ is precomputed offline for the LDPC code using the method described in Section~\ref{sec:VA}. The weighting coefficients in \eqref{eq:vulnerability_fusion} are set to $w_1=0$, $w_2=1$, $w_3=0.35$, and $w_4=0.25$. Note that $w_1=0$ because we use a regular LDPC code, hence all VNs have equal degree and degree provides no discriminative information. The extra energy allocated to nodes in the vulnerable set is set to $e=9$ in \eqref{eq:perturbation}, meaning nodes in $\mathcal{T}$ receive ten times the baseline perturbation weight. 

\textit{RL-based GMS attack (classical):} We implement the GMS attack using a DQN with a learning rate of $3\times10^{-4}$. The discount factor is $\gamma=0.99$, and the batch size is $64$. The replay buffer size is $10^5$. DQN is updated every $2048$ steps, and the total training timesteps is $10^5$. 

\textit{PGA attack (semantic):} The PGA attack uses a fixed step size $\alpha_t = 0.1$ for all iterations. The numerical stability constant for gradient normalization is $\varepsilon = 10^{-8}$. 

\textit{C\&W attack (semantic):} The C\&W attack is implemented as described in \eqref{e:CW} using the Adam optimizer with a learning rate of $0.01$. The initial regularization constant is $c=1.0$, with bounds $c_{\min}=10^{-6}$ and $c_{\max}=100$. For each sample, the optimization runs for up to $2000$ iterations or until the distortion target is met.

\begin{figure}[t]
    \centering
    \includegraphics[width=0.75\linewidth]{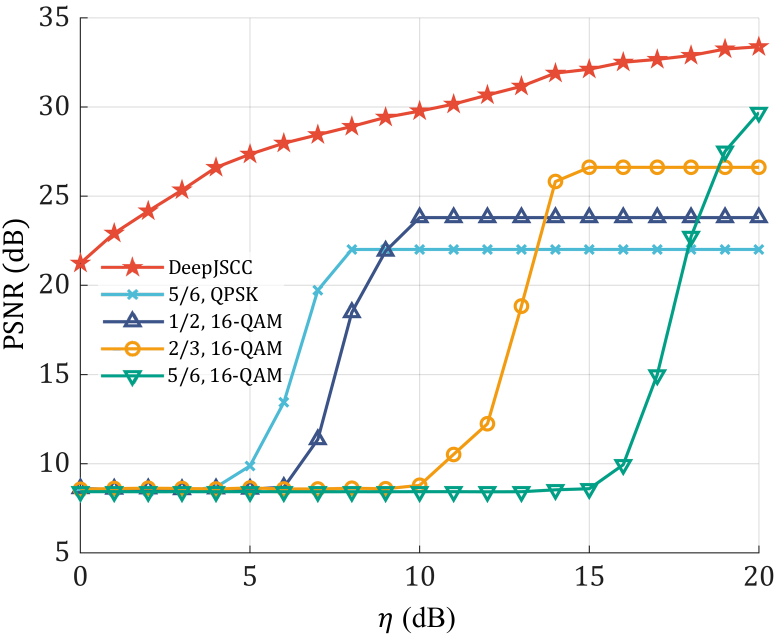}
    \caption{PSNR of reconstructed images for classical and semantic systems in the absence of adversarial attacks. The semantic system exhibits graceful degradation with SNR, while classical SSCC systems show a pronounced cliff effect.}
    \label{fig:chapter6_figure1}
\end{figure}

\subsubsection{Experimental results}
We begin by establishing the baseline performance without attacks. Figure~\ref{fig:chapter6_figure1} presents the PSNR versus SNR curves for both classical and semantic systems. For the classical SSCC schemes, we show results for four combinations of modulation and coding. As expected, the semantic communication system provides graceful degradation with SNR, maintaining reasonable reconstruction quality even at low SNRs. In contrast, the SSCC baselines exhibit a pronounced cliff effect: performance remains poor (dominated by decoding failures) below a certain SNR threshold, then rapidly transitions to a regime where performance is limited by compression and quantization distortion rather than channel errors. This observation, consistent with many prior works, serves as the reference baseline for the adversarial scenarios that follow.

\begin{figure}[t]
    \centering
    \includegraphics[width=0.75\linewidth]{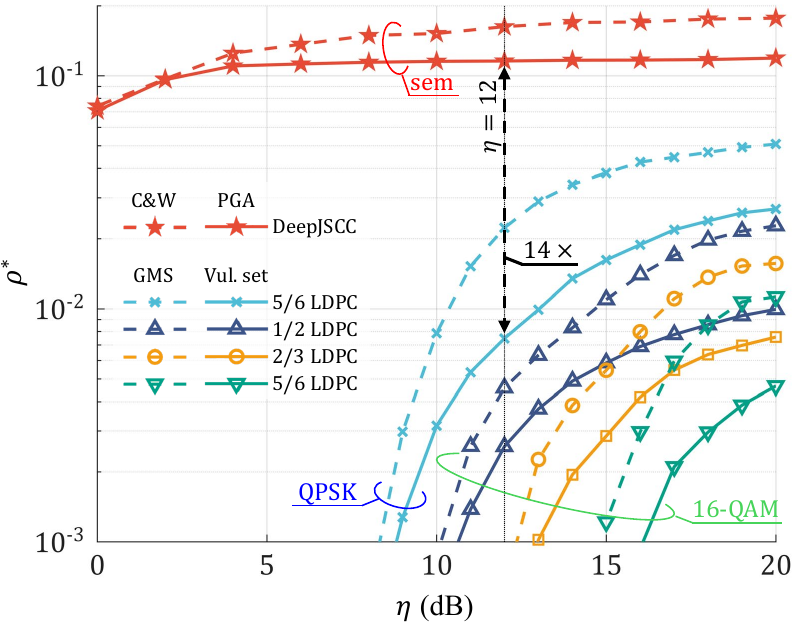}
    \caption{ Minimum attack power $\rho^*$ required to achieve the target distortion on the image transmission task for classical and semantic communication systems with different attack schemes. Lower $\rho^*$ indicates higher vulnerability. The semantic system consistently requires higher attack power than the classical system, demonstrating greater robustness.}
    \label{fig:chapter6_figure2}
\end{figure}

We now evaluate the minimum perturbation energy $\rho^*$ required to induce the target distortion ($\mathcal{D}^*$ corresponding to $15$dB PSNR) using the attack schemes proposed in Section~\ref{sec:V}. The complete experimental results are summarized in Figure~\ref{fig:chapter6_figure2}. Our main observations are as follows:
\begin{itemize}[leftmargin=0.5cm]
    \item For the classical system, the vulnerable set attack requires significantly lower attack power to achieve the target distortion compared to the RL-based GMS attack across all SNRs. For example, for QPSK and $5/6$-LDPC, the RL-based GMS attack achieves the target with $\rho^* \approx 2.2\times 10^{-2}$ at $12$dB SNR, whereas the vulnerable set attack requires only $\rho^* \approx 7.5\times 10^{-3}$, corresponding to a $3\times$ reduction in attack power. This validates that the structural insights from coding theory can be exploited to craft substantially more efficient perturbations than learning-based approaches.
    \item For the semantic system, the proposed PGA attack consistently requires much lower power than the C\&W attack across the entire SNR range. This demonstrates that our gradient-based approach, which directly optimizes for the minimal perturbation to reach the distortion target, is more powerful than the margin-based C\&W formulation adapted from classification tasks.
    \item Most importantly, the semantic communication system demonstrates significantly higher adversarial robustness than the classical SSCC system. Comparing the most effective attacks for each system (vulnerable set for classical, PGA for semantic), we observe that $\rho^*_{\text{sem},1}$ is consistently larger than $\rho^*{\text{sscc},1}$ across all SNRs. For instance, at $12$dB SNR, $\rho^*{\text{sem},1} \approx 0.11$ while $\rho^*_{\text{sscc},1} \approx 7.5\times 10^{-3}$. Thus, it is $14\times$ more difficult to attack semantic communication than the classical system. These empirical results provide strong evidence against the prevailing belief that semantic communication systems are inherently vulnerable.
\end{itemize}

\subsubsection{Ablation studies}
To further demonstrate the effectiveness of our vulnerable set attack, we perform an ablation study to investigate whether nodes from the vulnerable set indeed enhance attack efficiency. Figure~\ref{fig:chapter6_figure4} compares the minimum perturbation energy required to achieve the degradation target with and without extra energy allocated to vulnerable set nodes.
As shown, allocating extra perturbation energy to vulnerable set nodes consistently reduces the required attack energy across all modulation-coding settings,  with particularly pronounced gains at higher code rates. This confirms that the structural vulnerabilities identified offline translate into real-world attack efficiency gains, and that the adaptive thresholding scheme successfully identifies the most impactful nodes.

\begin{figure}[t]
    \centering
    \includegraphics[width=0.9\linewidth]{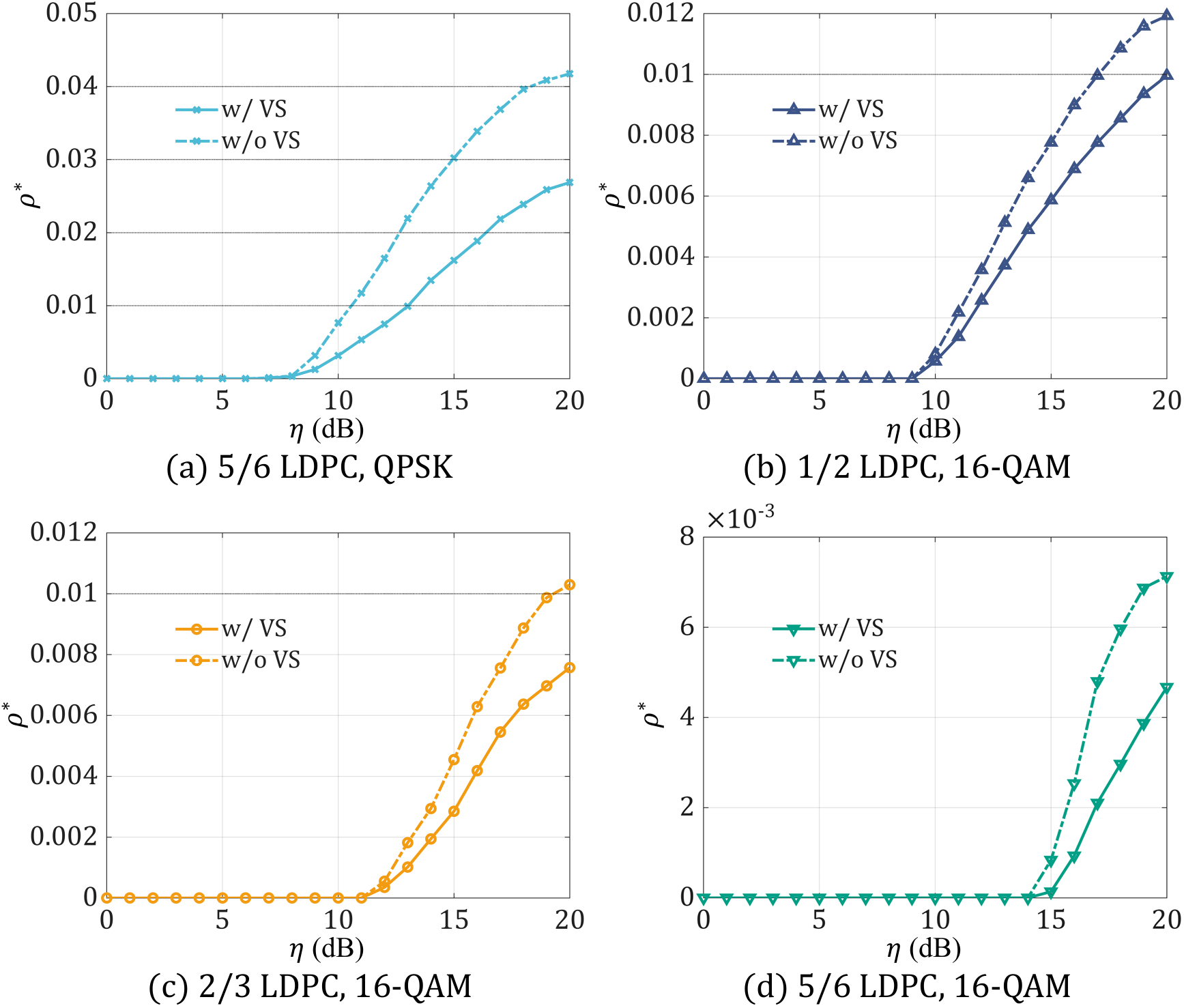}
    \caption{Ablation study on the impact of vulnerable node targeting in the vulnerable set attack for the image transmission task. Allocating extra perturbation energy to nodes in the vulnerable set (``w/ VS'') consistently reduces the required attack power compared to the baseline without targeting (``w/o VS'').}
    \label{fig:chapter6_figure4}
\end{figure}

\subsection{Massive MIMO CSI Transmission}

Next, we evaluate the adversarial robustness of classical and semantic systems in a massive MIMO CSI transmission task, a critical component in modern wireless systems where accurate channel knowledge at the base station (BS) is essential for beamforming and interference management. We consider a single-cell downlink massive MIMO system with $N_t$ transmit antennas at the BS and one receive antenna at the user equipment (UE). The CSI matrix $\bm{H} \in \mathbb{C}^{N_c \times N_t}$ stacks all channel vectors across $N_c$ OFDM subcarriers. To reduce feedback overhead from the UE to the BS, we sparsify $\bm{H}$ in the angular-delay domain using a 2D discrete Fourier transform (DFT):
\begin{equation}
    \bm{H} = \bm{F}_d \widetilde{\bm{H}} \bm{F}_a^H,
\end{equation}
where $\bm{F}_d$ and $\bm{F}_a$ denote the $N_c \times N_c$ and $N_t \times N_t$ DFT matrices corresponding to the delay and angular dimensions, respectively. The sparse representation $\widetilde{\bm{H}}$ enables significant compression of CSI, since most channel energy is concentrated in a few coefficients due to limited scattering clusters in massive MIMO systems. Only these significant coefficients need to be quantized and fed back to the BS, reducing the feedback overhead from $N_t \times N_c$ to a manageable size.

\subsubsection{System setup}
We generate channel matrices using the COST 2100 channel model \cite{liu2012cost} for the outdoor rural scenario at the $300$MHz band. This scenario has been widely utilized in CSI transmission tasks \cite{wen2018deep,guo2022overview}. The BS is located at the center of a $400,\mathrm{m}\times 400,\mathrm{m}$ square area, and the UEs are randomly deployed within this area for each sample. The BS employs a uniform linear array (ULA) with $N_t=32$ antennas and $N_c=1024$ subcarriers. When transforming the channel matrix into the angular-delay domain, we retain the first $32$ delay taps, resulting in $\widetilde{\bm{H}}$ of size $32 \times 32$. The training, validation, and testing sets contain $100,000$, $30,000$, and $20,000$ samples, respectively. All testing samples are excluded from the training and validation sets.

For the classical SSCC system, we use DFT-based compression that leverages the sparsity of the channel matrix in the angular domain. Specifically, we retain the $K$ largest-magnitude coefficients from $\widetilde{\bm{H}}$, where $K$ is chosen to match the bandwidth ratio. These coefficients are quantized using a non-uniform scalar quantizer and then protected by an LDPC code. For the results reported here, we use $K=256$ (equivalent to $25\%$ of the original $32\times32=1024$ coefficients, matching the bandwidth ratio of $1/4$ used in the image task) with 8-bit quantization, and LDPC coding with rates $1/2$, $2/3$, and $5/6$ for QPSK and 16-QAM, respectively.

\begin{figure}[t]
    \centering
    \includegraphics[width=0.75\linewidth]{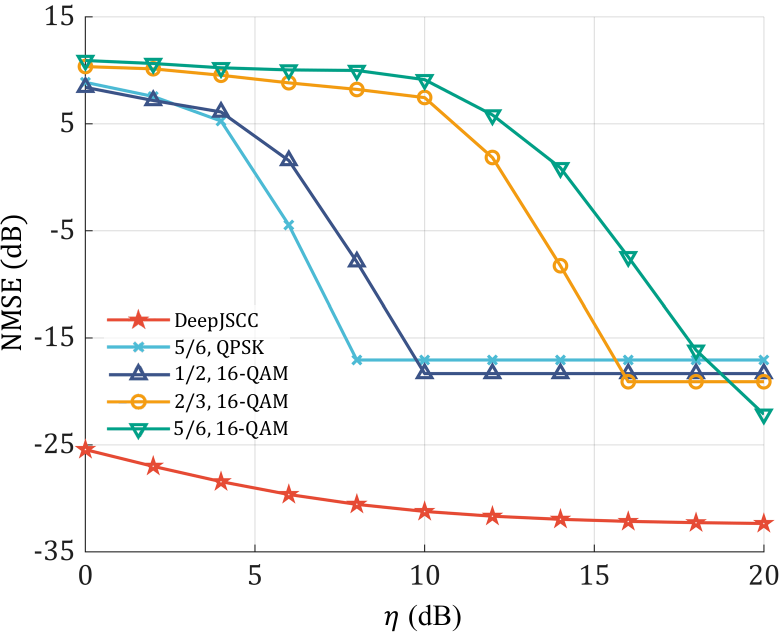}
    \caption{NMSE of reconstructed CSI matrices for classical and semantic systems in the absence of adversarial attacks. The semantic system achieves consistently lower NMSE across all SNRs, while classical systems exhibit a cliff effect.}
    \label{fig:chapter6_figure6}
\end{figure}

For the semantic communication system, we use the same CNN architecture as in the image transmission task, treating the $32\times 32$ angular-delay matrix as a single-channel image. The encoder maps this matrix to $N=256$ channel symbols (since $3072/4=768$ for images, but here the source dimension is $1024$, so $N=1024/4=256$). The decoder reconstructs the $32\times 32$ matrix from the noisy received symbols. Training uses the AdamW optimizer with a learning rate of $10^{-3}$, batch size of $512$, and MSE loss.

To evaluate reconstruction quality, we use the normalized MSE (NMSE) defined as
\begin{equation}
\text{NMSE}\text{ (dB)} = 10 \log_{10} \left( \frac{\| \bm{x} - \hat{\bm{x}} \|^2}{\| \bm{x} \|^2} \right).
\end{equation}
Lower NMSE values (more negative in dB) indicate better reconstruction performance. We set a target NMSE of $-16$dB for the CSI transmission task.

\subsubsection{Experimental results}

Analogous to the image transmission task, we first evaluate the reconstruction performance of both systems in the absence of attacks. Figure~\ref{fig:chapter6_figure6} shows the NMSE versus SNR curves. The semantic communication system achieves consistently lower NMSE across the entire SNR range, with graceful degradation. In contrast, the LDPC-based CSI transmission is highly sensitive to channel conditions: below a threshold SNR, decoding failures cause catastrophic degradation, with NMSE rising sharply. Above the threshold, performance is limited by the compression and quantization distortion floor.

\begin{figure}[t]
    \centering
    \includegraphics[width=0.75\linewidth]{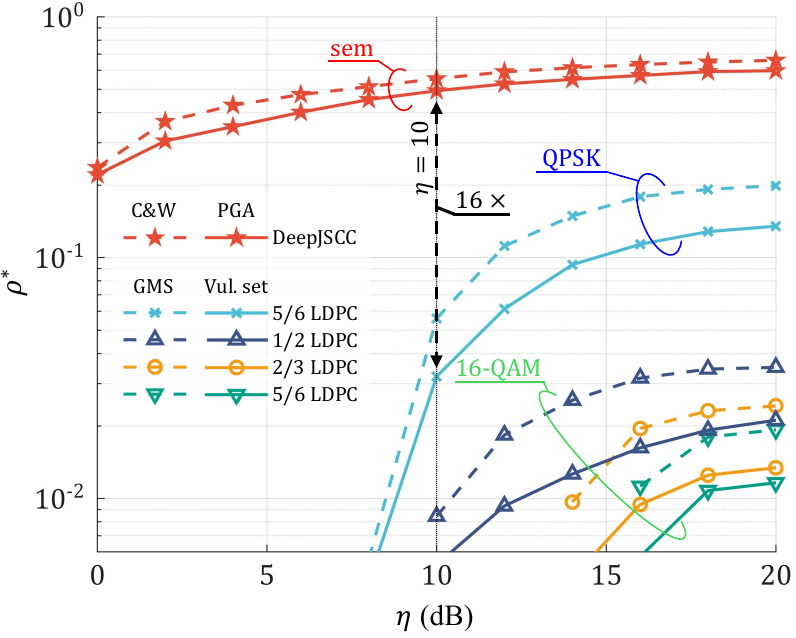}
    \caption{Minimum attack power $\rho^*$ required to achieve the target distortion on the massive MIMO CSI transmission task for classical and semantic communication systems with different attack schemes. The semantic system demonstrates substantially higher robustness, requiring up to $16\times$ more attack power than the classical system.}
    \label{fig:chapter6_figure3}
\end{figure}

Next, we introduce attacks following the methods in Section~\ref{sec:V} and evaluate the robustness of both systems. Figure~\ref{fig:chapter6_figure3} presents the minimum attack power $\rho^*$ as a function of SNR. The results mirror the conclusions drawn from the image transmission task:
\begin{itemize}[leftmargin=0.5cm]
\item For the classical SSCC system, the vulnerable set attack again proves substantially more efficient than the GMS attack, requiring up to $2\times$ less power to achieve the target NMSE across all SNRs. This confirms that the structured vulnerability identification is not task-specific but generalizes to different source statistics.
\item For the semantic system, PGA remains strictly more energy-efficient than C\&W, with the gap widening at lower SNRs.
\item The robustness disparity between the two systems is also pronounced in this task. The minimum attack power required for the semantic system consistently exceeds that for the classical system by a large margin. At $10$dB SNR, $\rho^*_{\text{sem},1} \approx 0.5$ while $\rho^*_{\text{sscc},1} \approx 3.1\times 10^{-2}$ (vulnerable set attack), meaning that semantic communication is more than $16\times$ difficult to attack than classical systems. This provides another strong piece of empirical evidence against the prevailing belief that semantic communication systems are inherently vulnerable.
\end{itemize}

\subsubsection{Ablation studies}

\begin{figure}[tb]
    \centering
    \includegraphics[width=0.9\linewidth]{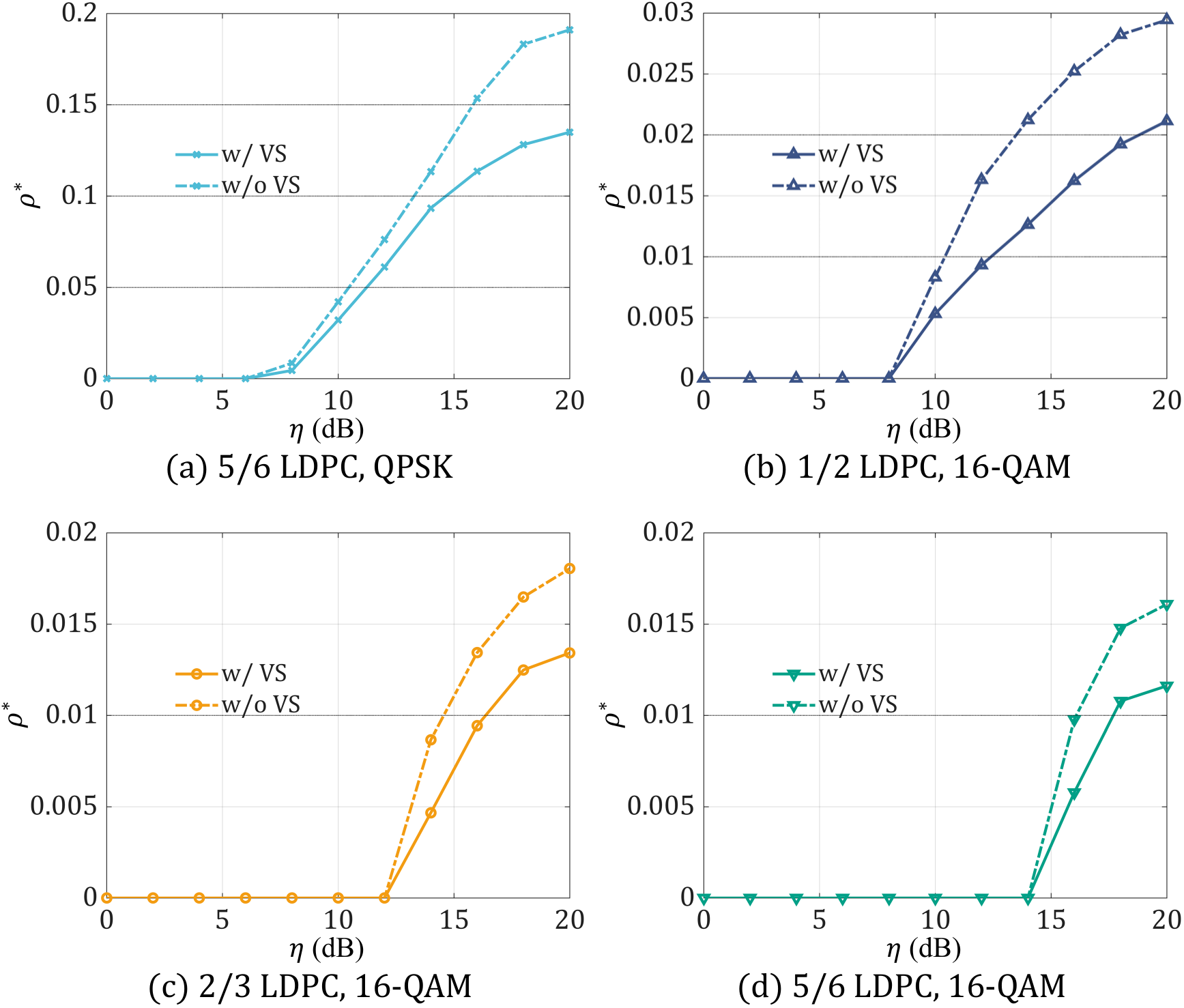}
    \caption{Ablation study on the impact of vulnerable node targeting in the vulnerable set attack for the CSI transmission task. Consistent with the image task, allocating extra energy to vulnerable set nodes significantly reduces the required attack power compared to the baseline.}
    \label{fig:chapter6_figure5}
\end{figure}

Finally, we investigate the effectiveness of the vulnerable set attack through ablation studies on the CSI task. Figure~\ref{fig:chapter6_figure5} compares the minimum perturbation energy with and without extra energy allocation to vulnerable set nodes.
The observations from Figure~\ref{fig:chapter6_figure5} corroborate the conclusions for image transmission. Across all configurations, attacks with extra perturbation energy allocated to vulnerable set nodes require less power than baseline attacks at the same SNR. This consistency across two very different tasks demonstrates that the benefit of vulnerable node targeting is not modality-specific; rather, it reflects a fundamental property of LDPC codes that can be exploited by an adversary regardless of the source statistics or modulation scheme.

\section{Conclusion}\label{sec:Conclusion}
This paper has challenged a deeply ingrained assumption in the emerging field of semantic communication: that its reliance on DL architectures renders it fundamentally more vulnerable to adversarial attacks than classical communication systems. Through rigorous theoretical analysis and extensive empirical validation, we have demonstrated a counterintuitive and consequential truth: semantic communication systems possess an inherent, unanticipated robustness to adversarial perturbations, which can even exceed that of their classical counterparts by more than an order of magnitude.

The implications, however, extend beyond this reversal. Our analysis exposes critical blind spots that demand further investigation. For semantic communication, the natural question becomes: can we deliberately enhance this implicit robustness? The Lipschitz constant of the decoder emerges as a key design parameter, suggesting that architectures explicitly regularized for smoothness, or trained with adversarial objectives, could push robustness even further. More ambitiously, can we characterize the fundamental trade-off between semantic efficiency and adversarial resilience, and design encoders that optimally balance both?

For classical systems, our vulnerable set attack reveals that decades of coding theory optimized for random noise have left structured codes dangerously exposed to intelligent adversaries. This opens an entirely new research frontier: the design of channel codes that are provably robust to targeted perturbations while maintaining near-capacity performance under benign conditions. Can we develop code families with tunable resilience, or decoding algorithms that detect and mitigate structured attacks? The path forward likely lies in hybrid approaches, i.e., leveraging the mathematical guarantees of classical codes while incorporating the adaptability of learning-based methods.

As wireless systems increasingly operate in contested environments and support safety-critical applications, adversarial robustness must transition from afterthought to first principle. This paper establishes that semantic communication, far from being a vulnerability, may offer a path toward that goal.

\appendices

\section{Proof of Lemma \ref{lem:D_sem_noise}}\label{App:A}
From the definition of $\mathcal{D}_{\text{sem},0}$, we have
\begin{eqnarray}
\hspace{-0.5cm}&&\hspace{-0.2cm} \mathcal{D}_{\text{Sem},0} 
= \mathbb{E}_{\bm{r}} \left[ \left\| g(\bm{r}; \bm{W}_g) - \bm{x}\right\|^2 \right] \notag\\
\hspace{-0.5cm}&& \stackrel{(a)}{\approx} \mathbb{E}_{\bm{z},\bm{\omega}}\left[ \left\| g\left(\left|h \right|\bm{z}; \bm{W}_g \right) + \bm{\mathcal{J}}_{ \! g}\left(\left|h \right|\bm{z} \right) \bm{\omega} - \bm{x} \right\|^2 \right] \notag\\
\hspace{-0.5cm}&& \stackrel{(b)}{=} \mathbb{E}_{\bm{z}}\left[(g(|h|\bm{z}; \bm{W}_g) - \bm{x})^\top \left(g\left(\left|h \right|\bm{z} ; \bm{W}_g \right) - \bm{x} \right) \right]  \notag\\
\hspace{-0.5cm}&& \hspace{2.75cm} + \mathbb{E}_{\bm{z},\bm{\omega}} \left[ \left(\bm{\mathcal{J}}_{ \! g} \left( \left|h \right|\bm{z} \right)\bm{\omega} \right)^\top \left(\bm{\mathcal{J}}_{ \! g} \left(\left|h \right|\bm{z} \right) \bm{\omega} \right) \right] \notag\\
\hspace{-0.5cm}&& \hspace{1.8cm} + 2\mathbb{E}_{\bm{z}} \left[  \left(g \left( \left|h \right|\bm{z}; \bm{W}_g \right) - \bm{x})^\top \bm{\mathcal{J}}_{ \! g}\left(|h|\bm{z} \right) \right] \mathbb{E}_{\bm{\omega}}[\bm{\omega} \right] \notag\\
\hspace{-0.5cm}&& \stackrel{(c)}{=} \mathbb{E}_{\bm{z}} \left[ \left(g \left( \left|h \right|\bm{z} ; \bm{W}_g\right) - \bm{x}\right)^\top \left( g \left( \left|h \right|\bm{z}; \bm{W}_g \right) - \bm{x} \right) \right]  \notag\\
\hspace{-0.5cm}&& \hspace{2.75cm} + \mathbb{E}_{\bm{z},\bm{\omega}}\left[ \left( \bm{\mathcal{J}}_{ \! g} \left( \left|h\right|\bm{z} \right)\bm{\omega} \right)^\top  \left(\bm{\mathcal{J}}_{ \! g} \left( \left|h \right|\bm{z}\right)\bm{\omega} \right) \right] \notag\\
\hspace{-0.5cm}&& \stackrel{(d)}{=}   \mathbb{E}_{\bm{z}}  \left[  \left(g  \left( \left|h \right|  \bm{z}; \bm{W}_g\right)   -  \bm{x} \right)^{\top}  \left(g  \left( \left|h \right|  \bm{z} ; \bm{W}_g\right)  -  \bm{x} \right) \right] \notag\\
\hspace{-0.5cm}&& \hspace{0.5cm} +  \sigma^2_{\omega} \mathbb{E}_{\bm{z}}  \left[  \sum_i  \sigma_{\bm{\mathcal{J}}_{  g},i}^2  \left( \left|h \right|  \bm{z} \right)  \right] 
= 
\sigma^2_\omega \mathbb{E}_{\bm{z}}  \left[  \sum_i  \sigma_{\bm{\mathcal{J}}_{ \! g},i}^2  \left( \left|h \right|  \bm{z} \right)  \right] \notag\\ 
\hspace{-0.5cm}&&
\stackrel{(e)}{\leq}
N \sigma^2_\omega G^2.        
    \end{eqnarray}
where (a) follows from \eqref{eq:taylor}, (b) exploits the independence between $\bm{z}$ and $\bm{\omega}$, (c) uses the fact that $\mathbb{E}_{\bm{\omega}}[\bm{\omega}] = 0$, (e) follows from the Lipschitz property of $g(\cdot)$, which implies that $\sum_i \sigma_i^2(\bm{\mathcal{J}}_g(|h|\bm{z}; \bm{W}_g)) \leq N G^2$, and (d) holds because 
\begin{eqnarray*}\label{e:N_attack}
\hspace{-0.3cm}& \hspace{-0.3cm}& \mathbb{E}_{\bm{z},\bm{\omega}}\Big[ \big( \bm{\mathcal{J}}_{g}(|h|\bm{z}) \bm{\omega} \big)^\top \big( \bm{\mathcal{J}}_{g}(|h|\bm{z}) \bm{\omega} \big) \Big]  \\
\hspace{-0.3cm}& = \hspace{-0.3cm}&  \mathbb{E}_{\bm{z},\bm{\omega}} \left[ \text{tr}\left\{ (\bm{\mathcal{J}}_{ \! g}(|h|\bm{z}))^\top \bm{\mathcal{J}}_{ \! g}(|h|\bm{z}) \bm{\omega} \bm{\omega}^\top \right\} \right] \\
\hspace{-0.3cm}& = \hspace{-0.3cm}& \text{tr}\left\{ \mathbb{E}_{\bm{z},\bm{\omega}} \left[ (\bm{\mathcal{J}}_{ \! g}(|h|\bm{z}))^\top \bm{\mathcal{J}}_{ \! g}(|h|\bm{z}) \bm{\omega} \bm{\omega}^\top \right] \right\} \\
\hspace{-0.3cm}& \overset{(d1)}{=} \hspace{-0.3cm}& \text{tr}\left\{ \mathbb{E}_{\bm{z}} \left[ (\bm{\mathcal{J}}_{ \! g}(|h|\bm{z}))^\top \bm{\mathcal{J}}_{ \! g}(|h|\bm{z}) \right] \mathbb{E}_{\bm{\omega}} \left[ \bm{\omega} \bm{\omega}^\top \right] \right\} \\
\hspace{-0.3cm}& = \hspace{-0.3cm}& \text{tr}\left\{ \mathbb{E}_{\bm{z}} \left[ (\bm{\mathcal{J}}_{ \! g}(|h|\bm{z}))^\top \bm{\mathcal{J}}_{ \! g}(|h|\bm{z}) \right] \sigma^2_\omega \bm{I}_N \right\} \\
\hspace{-0.3cm}& = \hspace{-0.3cm}&  \sigma^2_\omega \mathbb{E}_{\bm{z}} \left[ \| \bm{\mathcal{J}}_{ \! g}(|h|\bm{z}) \|^2 \right] \\
\hspace{-0.3cm}& = \hspace{-0.3cm}&  \sigma_\omega^2 \, \mathbb{E}_{\bm{z}}\Big[ \sum_i \sigma_{\bm{\mathcal{J}}_g,i}^2(|h|\bm{z}) \Big],
\end{eqnarray*}
in which (d1) follows from the statistical independence of the AWGN from all other signals and $\bm{I}_N$ is the $N \times N$ identity matrix. This completes the proof.

\section{Derivation of \eqref{e:s_t}}\label{App:B}

This appendix presents a detailed derivation of \eqref{e:s_t} using mathematical induction.
From the definitions in \eqref{e:recursion}, we have
\begin{eqnarray}\label{e:definition}
\bm{s}^{(t)} &\hspace{-0.6cm}=&\hspace{-0.6cm}
 \alpha^{(t)} \frac{\left. \nabla_{\bm{y}^{(t)}} D\bigl(\bm{x}, \widehat{\bm{x}}(\bm{y}^{(t)})\bigr) \right|_{\bm{y}^{(t)}=\bm{r}+\sum_{i=1}^{t-1} \bm{s}^{(i)}}}{\left\| \left. \nabla_{\bm{y}^{(t)}} D\bigl(\bm{x}, \widehat{\bm{x}}(\bm{y}^{(t)})\bigr) \right|_{\bm{y}^{(t)}=\bm{r}+\sum_{i=1}^{t-1} \bm{s}^{(i)}} \right\|} \\
&\hspace{-0.6cm} = &\hspace{-0.6cm} \alpha^{(t)} \bm{J}_g(\bm{y}^{(t)})^\top \big( g(\bm{y}^{(t)}; \bm{W}_g) - \bm{x} \big)
\nonumber \\
&\hspace{-0.6cm}\overset{\|\bm{s}^{(i)}\|\to 0}{=}&\hspace{-0.3cm} \! \alpha^{(t)}  \bm{J}_{\!g}( \bm{y}^{(t)} )^{\!\top} \!\! \left( \! g(\bm{y}^{(1)}; \bm{W}_{\!g}) \!+ \! \sum_{i=1}^{t-1} \bm{J}_g(\bm{y}^{(i)})\bm{s}^{(i)} \! - \! \bm{x} \! \right) \!.\nonumber
\end{eqnarray}

According to \eqref{e:s_t}, when $n=1$, we have:
\begin{equation}
\bm{s}^{(1)} = \alpha^{(1)} \left( \bm{\mathcal{J}}_{ \! g}(\bm{y}^{(1)}) \right)^\top \left( g(\bm{y}^{(1)}; \bm{W}_g) - \bm{x} \right),
\end{equation}
which satisfies \eqref{e:s_t}.

Assume the formula holds for all $n \leq t$, then for $n = t+1$, we have:
\begin{eqnarray}
\hspace{-0.5cm}&& \hspace{-0.2cm} \bm{s}^{(t+1)} \overset{\text{\eqref{e:definition}}}{=} \alpha^{(t+1)} \left( \bm{\mathcal{J}}_{  g}(\bm{y}^{(t+1)}) \right)^{\top} \\
\hspace{-0.5cm}&& \hspace{2.7cm}\left( g(\bm{y}^{(1)}; \bm{W}_g) - \bm{x} + \sum_{i=0}^{t} \bm{\mathcal{J}}_{  g}(\bm{y}^{(i)}) \bm{s}^{(i)} \right)  \notag\\
\hspace{-0.5cm}&& = \alpha^{(t+1)}  \left(  \bm{\mathcal{J}}_{  g}(\bm{y}^{(t+1)}) \right)^{ \top} \notag\\
\hspace{-0.5cm}&&\hspace{0.4cm} \left( g(\bm{y}^{(1)}; \bm{W}_g)  -  \bm{x}  +  \sum_{i=0}^{t-1} \bm{\mathcal{J}}_{  g}(\bm{y}^{(i)}) \bm{s}^{(i)}  +  \bm{\mathcal{J}}_{  g}(\bm{y}^{(t)}) \bm{s}^{(t)}  \right) 
 \notag\\
\hspace{-0.5cm}&& \hspace{-0.05cm} \stackrel{(a)}{=} \alpha^{(t+1)}  \left( \bm{\mathcal{J}}_{  g}(\bm{y}^{(t+1)}) \right)^{\top} \notag\\
\hspace{-0.5cm}&& \hspace{1cm} \Bigg(  g(\bm{y}^{(1)}; \bm{W}_g)  -  \bm{x}  +  \sum_{i=0}^{t-1} \bm{\mathcal{J}}_{  g}(\bm{y}^{(i)})  \bigg(  \alpha_{i}^\prime  \left( \bm{\mathcal{J}}_{  g}(\bm{y}^{(i)}) \right)^{\!\!\top} \notag\\
\hspace{-0.5cm}&& \hspace{-0.2cm}\left[ \prod_{j=0}^{i-1} \! \left( \! \bm{I}_M \! + \! \alpha^{(j)} \bm{\mathcal{J}}_{\!g}(\bm{y}^{(j)}) \! \left( \! \bm{\mathcal{J}}_{\!g}(\bm{y}^{(j)}) \right)^{ \!\top} \right) \! \right] \!\!  \left( \! g(\bm{y}^{(1)}; \bm{W}_g) \!-\! \bm{x} \! \right) \! \bigg)
 \notag\\
\hspace{-0.5cm}&&  +  \bm{\mathcal{J}}_{ g}  ( \bm{y}^{(t)} ) \alpha^{(t)} \left(  \bm{\mathcal{J}}_{ g}  (\bm{y}^{(t)})  \right)^{ \top}  \notag\\
\hspace{-0.5cm}&&  \hspace{-0.2cm}\left[ \prod_{i=0}^{t-1} \! \left( \! \bm{I}_{ M} \! + \! \alpha^{(i)} \bm{\mathcal{J}}_{ \! g}(\bm{y}^{(i)}) \! \left(\! \bm{\mathcal{J}}_{\! g}(\bm{y}^{(i)})  \right)^{ \! \top}  \right) \! \right] \! \left(\! g(\bm{y}^{(1)}; \bm{W}_g)  \!- \! \bm{x}\! \right) \!\! \Bigg) \notag\\
\hspace{-0.5cm}&& = \! \alpha^{\!(t+1)} \! \left( \! \bm{\mathcal{J}}_{ \! g}(\bm{y}^{(t+1)}) \! \right)^{\!\!\!\top} \!\!\! \left( \! \bm{I}_M \! + \! \sum_{i=0}^{t-1} \bm{\mathcal{J}}_{ \! g}(\bm{y}^{(i)}) \! \Bigg( \! \alpha^{(i)} \!\! \left( \! \bm{\mathcal{J}}_{ \! g}(\bm{y}^{(i)}) \!\right)^{\!\!\top} \!\! \right.\notag\\
\hspace{-0.5cm}&& \hspace{1.5cm}  \left[ \prod_{j=0}^{i-1}  \left(  \bm{I}_{M}  +  \alpha^{(j)} \bm{\mathcal{J}}_{  g}(\bm{y}^{(j)})  \left( \bm{\mathcal{J}}_{  g}(\bm{y}^{(j)}) \right)^{\top} \right)  \right]  \Bigg)  \notag\\
\hspace{-0.5cm}&&  +  \bm{\mathcal{J}}_{  g}(\bm{y}^{(t)})   \alpha^{(t)} \left( \bm{\mathcal{J}}_{  g}(\bm{y}^{(t)}) \right)^{\top} \notag\\
\hspace{-0.5cm}&&  \hspace{-0.2cm}\left. \left[ \prod_{i=0}^{t-1} \! \left( \! \bm{I}_M \! + \! \alpha^{\!(i)} \! \bm{\mathcal{J}}_{ \! g}(\bm{y}^{\!(i)}) \! \left(  \bm{\mathcal{J}}_{ \! g}(\bm{y}^{\!(i)}) \! \right)^{\!\!\top}  \right) \!\!  \right] \! \right)  \!\! \left( g(\bm{y}^{(1)}; \bm{W}_g) - \bm{x} \right)
 \notag\\
\hspace{-0.5cm}&& \hspace{-0.05cm} \stackrel{(b)}{=} \alpha^{(t+1)} \! \left( \! \bm{\mathcal{J}}_{ \! g}(\bm{y}^{(t+1)}) \! \right)^{\!\!\top} \! \left( \! \bm{I}_M \! + \! \bm{\mathcal{J}}_{ \! g}(\bm{y}^{(t)})   \alpha^{(t)} \! \left( \! \bm{\mathcal{J}}_{  g}(\bm{y}^{(t)}) \! \right)^{\!\!\top}  \right) 
 \notag\\
\hspace{-0.5cm}&& \left[ \prod_{i=0}^{t-1} \! \left( \! \bm{I}_M \! + \! \alpha^{\!(i)} \! \bm{\mathcal{J}}_{ \! g}(\bm{y}^{\!(i)}) \! \left(  \bm{\mathcal{J}}_{ \! g}(\bm{y}^{\!(i)}) \! \right)^{\!\!\top}  \right) \!\!  \right] \! \left( g(\bm{y}^{(1)}; \bm{W}_g) - \bm{x} \right)
 \notag\\
\hspace{-0.5cm}&&\hspace{-0.2cm} =   \alpha^{(t+1)}   \left(  \bm{\mathcal{J}}_{  g}(\bm{y}^{(t+1)})  \right)^{ \top}  \notag\\
\hspace{-0.5cm}&&   \left[ \prod_{i=0}^{t} \! \left( \! \bm{I}_M \! + \! \alpha^{\!(i)} \! \bm{\mathcal{J}}_{ \! g}(\bm{y}^{\!(i)}) \! \left(  \bm{\mathcal{J}}_{ \! g}(\bm{y}^{\!(i)}) \! \right)^{\!\!\top}  \right) \!\!  \right] \! \left( g(\bm{y}^{(1)}; \bm{W}_g) - \bm{x} \right), \notag  
\end{eqnarray}
where (a) uses the induction hypothesis for $n \leq t$, and (b) holds if the following equality is true:
\begin{eqnarray}
&&\hspace{-0.8cm} \bm{I}_{ M}  +  \sum_{i=0}^{t-1}  \bm{\mathcal{J}}_{  g}(\bm{y}^{(i)})  \Bigg(  \alpha^{(i)}  \left( \bm{\mathcal{J}}_{  g}(\bm{y}^{(i)}) \right)^{\top} \\
&&\hspace{1.5cm}  \prod_{j=0}^{i-1}  \left(  \bm{I}_{M}  +  \alpha^{(j)} \bm{\mathcal{J}}_{  g}(\bm{y}^{(j)})  \left( \bm{\mathcal{J}}_{  g}(\bm{y}^{(j)})  \right)^{\top}  \right) \Bigg) \notag\\
&&\hspace{1.4cm} = \prod_{i=0}^{t-1} \left( \bm{I}_M + \alpha^{(i)} \bm{\mathcal{J}}_{  g}(\bm{y}^{(i)}) \left( \bm{\mathcal{J}}_{  g}(\bm{y}^{(i)}) \right)^\top \right).\notag
\end{eqnarray}
We now prove this equality by induction on $l$.

For the base case $l=1$:
\begin{eqnarray}
&&\hspace{-0.7cm}\text{LHS} = \bm{I}_M + \sum_{i=0}^{0} \bm{\mathcal{J}}_{ \! g}(\bm{y}^{(i)}) \Bigg( \alpha^{(i)} \left( \bm{\mathcal{J}}_{ \! g}(\bm{y}^{(i)}) \right)^\top\\
&&\hspace{-0.7cm} \left[ \prod_{j=0}^{i-1}\! \left( \! \bm{I}_M \! + \! \alpha^{(j)} \bm{\mathcal{J}}_{ \! g}(\bm{y}^{(j)}) \left( \bm{\mathcal{J}}_{ \! g}(\bm{y}^{(j)}) \!\right)^{\!\!\top}  \right) \! \right] \! \Bigg) = \bm{I}_M = \text{RHS}. \notag
\end{eqnarray}
Assume the equality holds for $l \leq t$, then for $l = t+1$:
\begin{eqnarray}
\hspace{-0.6cm} && \text{LHS}
=  \bm{I}_{ M}  +  \sum_{i=0}^{t}  \bm{\mathcal{J}}_{\!g}  (\bm{y}^{(i)}) \Bigg(  \alpha^{(i)} \left(  \bm{\mathcal{J}}_{\!g}  (\bm{y}^{(i)})  \right)^{\top}  \notag\\
\hspace{-0.6cm} && \hspace{ 0.3cm}  \left[ \prod_{j=0}^{i-1}  \left(  \bm{I}_{M}  +  \alpha^{(j)} \bm{\mathcal{J}}_{\!g}(\bm{y}^{(j)})  \left( \bm{\mathcal{J}}_{\!g}(\bm{y}^{(j)})  \right)^{\top}  \right)  \right]  \Bigg) \notag\\
\hspace{-0.6cm} && =  \bm{I}_{ M}  +  \sum_{i=0}^{t-1}  \bm{\mathcal{J}}_{\!g}  (\bm{y}^{(i)})  \Bigg(  \alpha^{(i)} \left(  \bm{\mathcal{J}}_{\!g}  (\bm{y}^{(i)})  \right)^{\top}   \notag\\
\hspace{-0.6cm} && \hspace{ 0.5cm} \prod_{j=0}^{i-1}  \left(  \bm{I}_{M}  +  \alpha^{(j)} \bm{\mathcal{J}}_{\!g}(\bm{y}^{(j)})  \left( \bm{\mathcal{J}}_{\!g}(\bm{y}^{(j)})  \right)^{\top}  \right)  \Bigg)  + \bm{\mathcal{J}}_{\!g}(\bm{y}^{(t)}) \notag\\
\hspace{-0.6cm} &&  \left( \!  \alpha^{(t)} \! \left( \! \bm{\mathcal{J}}_{\!g}(\bm{y}^{(t)}) \!  \right)^{\! \top} \!  \left[ \prod_{j=0}^{t-1} \! \left( \! \bm{I}_M \! + \! \alpha^{(j)} \! \bm{\mathcal{J}}_{\!g}(\bm{y}^{(j)}) \!\! \left( \bm{\mathcal{J}}_{\!g}(\bm{y}^{(j)}) \right)^{\!\!\top} \right) \! \right] \! \right)  \notag\\
\hspace{-0.6cm} && \hspace{-0.05cm} \stackrel{(c1)}{=} \prod_{i=0}^{t-1}  \left( \bm{I}_M + \alpha^{(i)} \bm{\mathcal{J}}_{\!g}(\bm{y}^{(i)})  \left( \bm{\mathcal{J}}_{\!g}(\bm{y}^{(i)}) \right)^{\top} \right)  + \bm{\mathcal{J}}_{\!g}(\bm{y}^{(t)})   \alpha^{(t)}    \notag\\
\hspace{-0.6cm} && \hspace{0.2cm} \left( \bm{\mathcal{J}}_{\!g}(\bm{y}^{(t)}) \right)^{\top}  \left[ \prod_{j=0}^{t-1} \left( \bm{I}_M + \alpha^{(j)} \bm{\mathcal{J}}_{\!g}(\bm{y}^{(j)}) \left( \bm{\mathcal{J}}_{\!g}(\bm{y}^{(j)}) \right)^{\top}  \right)  \right] \notag\\
\hspace{-0.6cm} && =  \left(  \bm{I}_{M}  +  \bm{\mathcal{J}}_{\!g}(\bm{y}^{(t)}) \alpha^{(t)}  \left( \bm{\mathcal{J}}_{\!g}(\bm{y}^{(t)})  \right)^{\top}  \right)  \notag\\
\hspace{-0.6cm} && \hspace{2.7cm} \prod_{i=0}^{t-1}  \left( \bm{I}_{M}  +  \alpha^{(i)} \bm{\mathcal{J}}_{\!g}(\bm{y}^{(i)})  \left( \bm{\mathcal{J}}_{\!g}(\bm{y}^{(i)})  \right)^{\top}  \right)  \notag\\
\hspace{-0.6cm} && = \prod_{i=0}^{t} \left( \bm{I}_M + \alpha^{(i)} \bm{\mathcal{J}}_{\!g}(\bm{y}^{(i)}) \left( \bm{\mathcal{J}}_{\!g}(\bm{y}^{(i)}) \right)^\top \right) = \text{RHS},
\end{eqnarray}
where (c1) uses the induction hypothesis for $l=t$. Therefore, the equality holds for $l=t+1$, completing the induction.

Thus, the formula \eqref{e:s_t} holds for $n=t+1$, proving the general expression.

\bibliographystyle{IEEEtran}
\bibliography{./References.bib}

\end{document}